\documentclass[lettersize,journal]{IEEEtran}
\usepackage{amsmath,amsfonts}
\usepackage{algorithmic}
\usepackage{algorithm}
\usepackage{array}
\usepackage[caption=false,font=normalsize,labelfont=sf,textfont=sf]{subfig}
\usepackage{textcomp}
\usepackage{stfloats}
\usepackage{url}
\usepackage{verbatim}
\usepackage{graphicx}
\usepackage{cite}

\usepackage{mathrsfs}
\usepackage{cases}
\usepackage{bm}
\usepackage{makecell}
\usepackage{amsthm}
\newtheorem{theorem}{Theorem}
\allowdisplaybreaks[4]

\usepackage{color}

\begin{document}

\title{Latency-Constrained Resource Synergization for Mission-Oriented 6G Non-Terrestrial Networks}

\author{Yueshan Lin, Wei Feng, \IEEEmembership{Senior Member, IEEE}, Yunfei Chen, \IEEEmembership{Fellow, IEEE}, Yongxu Zhu, \IEEEmembership{Senior Member, IEEE}, Ning Ge, \IEEEmembership{Member, IEEE}, Shi Jin, \IEEEmembership{Fellow, IEEE}
        % <-this % stops a space
\thanks{Yueshan Lin, Wei Feng, and Ning Ge are with the Department of Electronic Engineering, State Key Laboratory of Space Network and Communications, Tsinghua University, Beijing 100084, China (e-mail: lin-ys17@tsinghua.org.cn; fengwei@tsinghua.edu.cn; gening@tsinghua.edu.cn).}% <-this % stops a space
\thanks{Yunfei Chen is with the Department of Engineering, University of Durham, Durham, U.K. DH1 3LE (e-mail: Yunfei.Chen@durham.ac.uk).}
\thanks{Yongxu Zhu and Shi Jin are with the National Mobile Communications Research Laboratory, Southeast University, Nanjing 210096, China. (e-mail: yongxu.zhu@seu.edu.cn; jinshi@seu.edu.cn).}
% \thanks{Manuscript received April 19, 2021; revised August 16, 2021.}
}

% The paper headers
%\markboth{Journal of \LaTeX\ Class Files,~Vol.~14, No.~8, August~2021}%
%{Shell \MakeLowercase{\textit{et al.}}: A Sample Article Using IEEEtran.cls for IEEE Journals}

%\IEEEpubid{0000--0000/00\$00.00~\copyright~2021 IEEE}
% Remember, if you use this you must call \IEEEpubidadjcol in the second
% column for its text to clear the IEEEpubid mark.

\maketitle

\begin{abstract}
This paper investigates latency-constrained resource synergization for mission-oriented non-terrestrial networks (NTNs) in post-disaster emergency scenarios. When terrestrial infrastructures are damaged, unmanned aerial vehicles (UAVs) equipped with edge information hubs (EIHs) are deployed to provide temporary coverage and synergize communication and computing resources for rapid situation awareness. We formulate a joint resource configuration and location optimization problem to minimize overall resource costs while guaranteeing stringent latency requirements. Through analytical derivations, we obtain closed-form optimal solutions that reveal the fundamental tradeoff between communication and computing resources, and develop a successive convex approximation method for EIH location optimization. Simulation results demonstrate that the proposed scheme achieves approximately 20\% cost reduction compared with benchmark approaches, validating its optimality and effectiveness for mission-critical emergency response applications in the sixth-generation (6G) era.
\end{abstract}


\section{Introduction}
Natural disasters (floods, earthquakes, fires, {\it etc.}) pose a great threat to human lives and the economy \cite{intro 01}. To reduce the losses caused by disasters, rapid situation awareness of the disastrous areas is of great significance. This requires that the data collected by field sensors be transmitted to the emergency command center with low latency \cite{intro 02}. However, terrestrial communication infrastructures could be damaged in such scenarios \cite{intro 03}. Disasters may also take place in rural or remote areas (forests, oceans, {\it etc.}) without terrestrial network coverage \cite{intro 04}. Therefore, non-terrestrial networks (NTNs), such as satellites and unmanned aerial vehicle (UAV), need to be utilized to provide coverage in disastrous areas \cite{intro 05} \cite{intro s01}.

Low-latency sensing data uploading brings many challenges to NTN system designs. Specifically, the communication capabilities of UAVs and satellites are inherently limited due to the size, weight and power constraints \cite{NTN} \cite{intro s02}. When high-resolution sensors generate a large amount of data, the NTN’s communication capability might not satisfy the latency requirements for data uploading. To resolve this problem, mobile edge computing (MEC) is a prospective paradigm \cite{MEC}. By processing the sensing data at the network edge, the key information of the data could be extracted before uploading, which reduces the overall latency. 

However, MEC-empowered NTNs lead to mobile infrastructure deployment, complicated data scheduling and multi-dimensional network resource orchestration. The number of its possible network states is considerably larger than traditional cellular networks, making it difficult for proper and efficient system design. To address this concern, an edge information hub (EIH) has been considered to incorporate and synergize the communication and computing resources in the NTN \cite{EIH RW sat-UAV MEC 05}. The information contained in the sensing data is extracted, processed and transmitted at the EIH for rapid situation awareness. Compared with traditional MEC networks, an EIH-empowered network enables the cooperative design of the communication and computing process to improve the network efficiency. Besides, basic structures and relationships within an EIH-empowered system could be explored to ease the system design.

In this paper, we investigate through analytical tools the fundamental tradeoff between communication and computing resource configuration in an EIH-empowered NTN. Building upon this tradeoff, we propose a joint resource configuration and location optimization scheme for the EIH to achieve efficient data uploading. Compared with the previous work \cite{EIH RW sat-UAV MEC 05} that optimizes resource orchestration and data scheduling, the key novelty of this proposed scheme lies in determining both the exact amount of resources deployed on the EIH and its optimal location, thereby meeting the uploading latency requirement while minimizing the overall resource cost. Notably, resource configuration and resource orchestration are jointly considered in the optimization, as the former is inherently affected by how the resources are allocated. The proposed joint optimization scheme ensures timely completion of data uploading while preventing excessive resource usage with low efficiency. The main contributions of this paper are summarized as follows.

\IEEEpubidadjcol
\begin{itemize}
	\item We propose an EIH-empowered NTN model for low-latency data uploading. Using this model, we formulate a joint resource configuration and position optimization problem to minimize the overall resource cost, under the constraint of uploading latency. The considered resources include the user-UAV transmission bandwidth, the UAV-satellite data rate and EIH's computing capability. The orchestration of these resources are jointly considered in the optimization problem.
	\item We apply a two-step strategy to solve the problem. First, the resource configuration is optimized with fixed EIH location. To address the complicated piecewise functions in the optimization problem, we provide a closed-form optimal solution through analytical derivations. Based on this closed-form solution, we illustrate the basic tradeoff between the EIH's communication and computing resource configuration.
	\item We then utilize the successive convex approximation method to optimize the EIH's location based on the closed-form expressions of optimal resource configuration. On this basis, we propose a joint resource configuration and location optimization scheme. Simulation results corroborate our theoretical achievements, and also demonstrate the superiority of our proposed scheme in reducing the overall resource cost.
\end{itemize}

The rest of this paper is as follows. In Section II, we review the related works on non-terrestrial networks and edge computing. The system model of the EIH-empowered NTN is introduced in Section III. In Section IV, we formulate the joint resource configuration and location optimization problem for cost minimization. We also provide our solution in this section. The simulation results are presented in Section V, and we draw the conclusion in Section VI. 

\section{Related Works}
\hspace{3.2mm} {\it 1) Non-Terrestrial Networks:} A number of existing studies have explored utilizing NTNs to provide communication services in rural or remote areas. Fang {\it et al.} \cite{RW NTN 01} offered a novel perspective that a complex NTN system can be regarded as multifarious combinations of three basic models, and further discussed the state-of-the-art technologies and future research directions for each model. More specifically, Wei {\it et al.} \cite{RW NTN s01} comprehensively reviewed the scenarios and key technologies of UAV-assisted data collection, addressing major issues including sensor clustering, UAV data collection modes, as well as joint path planning and resource allocation. Bian {\it et al.} \cite{RW NTN s02} considered utilizing a single UAV for data collection under multi-jammer attacks, where the collection schedule, power control, and UAV trajectory were jointly optimized to improve the uploading data rate. Li {\it et al.} \cite{RW NTN s03} further investigated a multi-UAV scenario, jointly designing sensor-UAV association and UAV trajectories to minimize task completion time. Liu {\it et al.} \cite{RW NTN 03} explored employing a UAV swarm to provide network coverage, proposing a power and subchannel allocation scheme to enhance network efficiency. Moreover, some studies have considered integrated satellite-UAV networks for communication \cite{RW NTN 04} \cite{RW NTN 05}. Specifically, Li {\it et al.} \cite{RW NTN 04} maximized system energy efficiency through optimizing UAV deployment, subchannel allocation, and power control. Ma {\it et al.} \cite{RW NTN 05} investigated a complex network comprising multiple UAVs and satellites, jointly designing satellite-UAV association, UAV trajectories, and power and bandwidth allocation based on successive convex approximation and block coordinate descent (BCD) techniques to maximize data gathering efficiency.

{\it 2) MEC-Empowered UAV Networks:} Due to the limited capacity of NTNs, many studies have considered the use of MEC technologies to empower the network. A major part of the studies focused on MEC-empowered UAV networks. For instance, Zhang {\it et al.} \cite{RW UAV MEC 01} explored the joint design of UAV trajectory, data scheduling decision, power usage and computing resource allocation in a single-UAV network. They proposed a three-step solution framework to minimize the weighted sum of task latency and energy consumption. In \cite{RW UAV MEC 02}, Wang {\it et al.} utilized a UAV to provide energy for users wirelessly and meanwhile collect their data, where the UAV trajectory, data scheduling and energy harvesting time are jointly optimized for energy consumption minimization. Zhang {\it et al.} \cite{RW UAV MEC 03} further considered a multi-UAV network, and the user association is jointly optimized with the UAV deployment and task data allocation to minimize the average response delay. In \cite{RW UAV MEC 04}, Wang {\it et al.} proposed an energy-oriented optimization scheme for multi-UAV trajectory, user association and computing resource allocation based on deep reinforcement learning. Liu {\it et al.} \cite{RW UAV MEC 05} considered optimizing the number of UAVs jointly with the UAV deployment and data scheduling decision for latency minimization. In \cite{RW UAV MEC 06}, Seid {\it et al.} investigated an integrated network with coexisting UAVs and ground base stations, where the user association and the communication and computing resources are jointly designed to reduce both latency and energy consumption. Zhao {\it et al.} \cite{RW UAV MEC 07} proposed a cooperative multi-agent deep-reinforcement learning framework to jointly design the trajectories, task data allocation, and resource management of UAVs.

{\it 3) MEC-Empowered Integrated Satellite-UAV Networks:} A number of existing studies further investigated MEC-empowered integrated satellite-UAV networks. Lin {\it et al.} \cite{RW sat-UAV MEC 01} proposed three minimal structures of an integrated satellite-MEC network and discussed the characteristics and key problems for each structure. Many existing studies focused on the problems of data scheduling and network resource allocation. For instance, Zhou {\it et al.} \cite{RW sat-UAV MEC 02} considered the task data scheduling of a single MEC-enabled UAV, and they proposed a latency-oriented algorithm based on deep reinforcement learning. Tun {\it et al.} \cite{RW sat-UAV MEC 03} utilized collaborative MEC-enabled UAVs for data offloading, where the data scheduling decision is optimized for latency minimization. In \cite{RW sat-UAV MEC 04}, Ei {\it et al.} jointly optimized the data scheduling and the bandwidth allocation to minimize the task completion latency under energy consumption constraints. Assuming concurrent communication and computing during data uploading, Lin {\it et al.} \cite{EIH RW sat-UAV MEC 05} proposed an optimal data scheduling and resource allocation scheme to further reduce the uploading latency.

Moreover, some studies focused on designing the UAVs’ positions or trajectories in an integrated satellite-UAV network. For instance, Chao {\it et al.} \cite{RW sat-UAV MEC 06} jointly designed the UAV placement and the data offloading decision to maximize the MEC service profit. Considering different conditions of satellite connection, Jung {\it et al.} \cite{RW sat-UAV MEC 07} proposed a joint UAV trajectory and data scheduling scheme to reduce system energy consumption with task latency constraints. In \cite{RW sat-UAV MEC 08}, Qi {\it et al.} further considered the user association problem jointly with UAV deployment and task offloading decision, where an iterative algorithm was proposed for latency minimization. In \cite{RW sat-UAV MEC 09}, Tong {\it et al.} jointly optimized the user clustering and the locations of UAV collection points to minimize the age of information of collected data. Moreover, Huang {\it et al.} \cite{RW sat-UAV MEC 10} considered a complex network consisting of multiple UAVs and satellites. They investigated the optimal association among users, UAVs and satellites, along with UAV trajectory and computing resource allocation, to minimize the energy consumption.

Despite the progress made so far, there still exist many research gaps in the system design of edge-computing-empowered NTNs. Specifically, most existing studies focused on the resource allocation problem, while the resource configuration has not been considered. Configuring proper amount of communication and computing resources in the network guarantees the data uploading task be completed within the latency constraint, while preventing excessive resource usage. In this paper, we focus on an EIH-empowered network, where the communication, computing and storage resources are synergized at the EIH. Through analytical derivations, we investigate the basic tradeoff between the network resources. On this basis, we propose the closed-form expressions of the optimal communication and computing resource configuration to minimize the overall resource cost.

\begin{figure*}[t]
	\centering
	\includegraphics[height=3.4in]{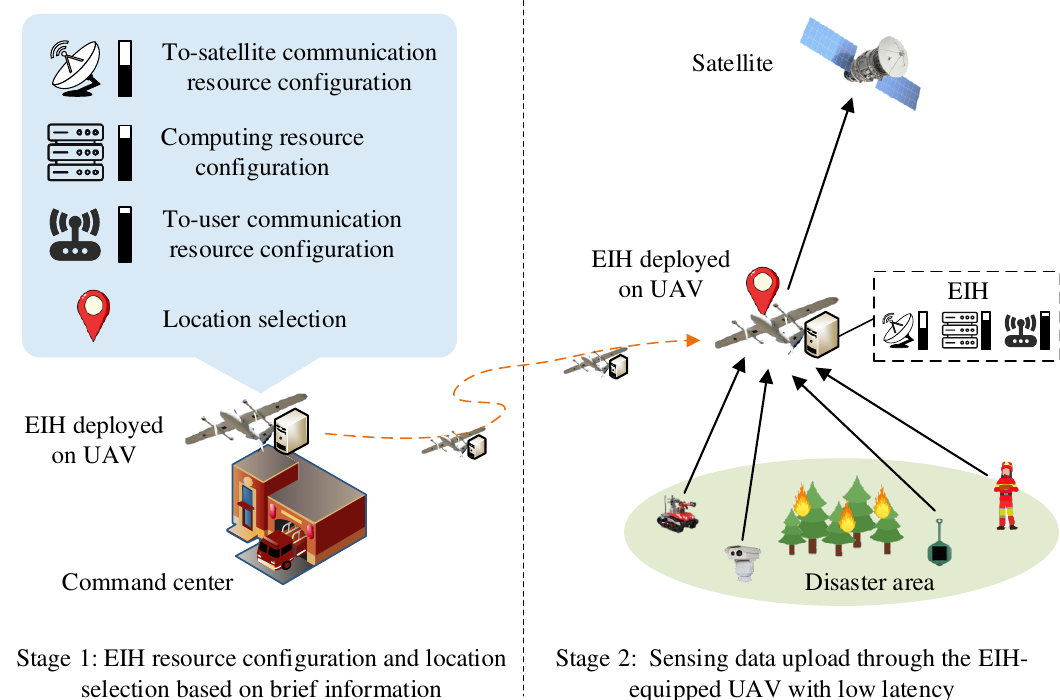}%
	\caption{EIH-empowered non-terrestrial network used for sensing data uploading in disaster areas.}
	\label{fig syst mod}
\end{figure*}

\begin{table}[t]
	\caption{List of Symbols}
	\centering
	\renewcommand{\arraystretch}{1.2}
	\begin{tabular}{|m{0.6in}|m{2.4in}|}
		\hline
		Symbol & Definition \\ \hline
		$U$ & Number of users \\ \hline
		$\mathcal{U}$ & Set of users \\ \hline
		$(h_{u,x},h_{u,y})$ & Horizontal coordinate of user $u$ \\ \hline
		$D_u$ & Amount of user $u$'s task data \\ \hline
		$T_{\rm req}$ & Overall latency requirement of data uploading \\ \hline
		$(h_x,h_y)$ & Horizontal coordinate of the EIH \\ \hline
		$H$ & Height of the EIH \\ \hline
		$B_{\rm total}$ & Total configured bandwidth for EIH’s to-user communication module \\ \hline
		$R^S_{\rm total}$ & Total configured data rate for EIH’s to-satellite communication module \\ \hline
		$F_{\rm total}$ & Total configured CPU frequency for EIH’s computing module \\ \hline
		$V_{\rm total}$ & Total configured storage for EIH’s storage module \\ \hline
		$B_u$ & User $u$'s allocated bandwidth for user-EIH transmission \\ \hline
		$\eta_u$ & Ratio of user $u$'s to-be-computed task data \\ \hline
		$F_u$ & User $u$'s allocated CPU frequency \\ \hline
		$\rho_u$ & Computing intensity of user $u$'s task data \\ \hline
		$\zeta_u$ & Computing output-to-input ratio of user $u$'s task data \\ \hline
		$R_u^S$ & User $u$'s allocated data rate for EIH-satellite transmission \\ \hline
		$g_u$ & Channel gain of the transmission between user $u$ and the EIH \\ \hline
		$s_u$ & Small-scale channel gain of the transmission between user $u$ and the EIH \\ \hline
		$l_u$ & Large-scale channel gain of the transmission between user $u$ and the EIH \\ \hline
		$c$ & Light speed \\ \hline
		$f$ & Carrier frequency of user-EIH transmission \\ \hline
		\makecell{$(a, b, \eta_{\rm LoS},$ \\ $\ \ \ \ \eta_{\rm NLoS})$} & Large-scale channel parameters related to the propagation environment \\ \hline
		$d_u$ & Distance between user $u$ and the EIH \\ \hline
		$\theta_u$ & Elevation angle between user $u$ and the EIH \\ \hline
		$R_u$ & Data rate of the transmission between user $u$ and the EIH \\ \hline
		$p_u$ & User $u$'s transmit power \\ \hline
		$\sigma^2$ & Noise power \\ \hline
		$T_u$ & Overall uploading latency for user $u$ \\ \hline
		$V_u$ & Minimum required storage for user $u$ \\ \hline
		\makecell{$(a_1, a_2, \ \ \ $ \\ $\ \ \ a_3, a_4)$} & Weight parameters of the overall resource costs \\ \hline
	\end{tabular}
	\label{Notation}
\end{table}

\section{System Model}
In Fig. \ref{fig syst mod}, we show how the EIH-empowered UAV network helps achieve situation awareness in a post-disaster emergency response scenario. In the first stage, an EIH-equipped UAV is prepared at the command center after disasters occur. To meet the uploading latency requirement and meanwhile avoid resource overuse, an appropriate amount of communication and computing resources need to be configured on the EIH, and the location of the EIH-equipped UAV needs to be properly selected. Note that such decisions need to be made only based on brief information on the disaster area ({\it e.g.}, sensors' distribution and data volume). In the second stage, the properly-configured EIH-equipped UAV is dispatched to the selected location and helps upload the sensing data through satellite backhaul.

We assume that there are $U$ sensors (denoted as $\mathcal{U} = \{1,2,...,U\}$). Sensor $u$ is located at coordinate $(h_{u,x},h_{u,y},0)$, which is assumed fixed and known. The amount of sensor $u$'s sensing data is denoted as $D_u$. All sensor data need to be uploaded within latency less than $T_{\rm req}$. The EIH-equipped UAV is located at $(h_{x},h_{y},H)$, where the height $H$ is fixed whereas the horizontal coordinate $(h_{x},h_{y})$ can be adjusted. The EIH incorporates two communication modules for to-user and to-satellite transmissions, a computing module (\textit{i.e.}, an MEC server) to process the data and extract key information, and a storage module to cache the data pending for transmission or processing. We assume that the EIH’s to-user communication module is configured with bandwidth $B_{\rm total}$, and the to-satellite communication module is configured with data rate $R^S_{\rm total}$. Besides, the EIH’s computing module is assumed to be configured with CPU frequency $F_{\rm total}$, and the total configured storage is denoted as $V_{\rm total}$. The overall system variables and parameters are summarized in Table \ref{Notation}.

As shown in Fig. \ref{fig data flow}, we assume that the transmissions and computing proceed concurrently for lower uploading latency. Specifically, users transmit their data to the EIH and the data is stored in its storage module. We assume that the user-EIH transmissions adopt the frequency division multiple access (FDMA) scheme, where the allocated bandwidth for user $u$ is denoted as $B_u$. The data of user $u$ could be divided into two parts, namely to-be-computed data (ratio $\eta_u$) and to-be-uploaded data (ratio $1-\eta_u$). The to-be-computed data are processed in the EIH's computing module, and the outcome are returned to the storage module as to-be-uploaded data. The to-be-uploaded data (including both the original to-be-uploaded data and the computation outcome) are transmitted to the satellite utilizing the EIH's to-satellite terminal. For user $u$, the allocated computing resources are $F_u$, the computing intensity is $\rho_u$ (in CPU cycles/bit), the computing output-to-input ratio is $\zeta_u$, and the allocated to-satellite data rate is $R_u^S$. Note that these communication and computing processes are carried out concurrently instead of sequentially.

\begin{figure}[t]
	\centering
	\includegraphics[height=2.4in]{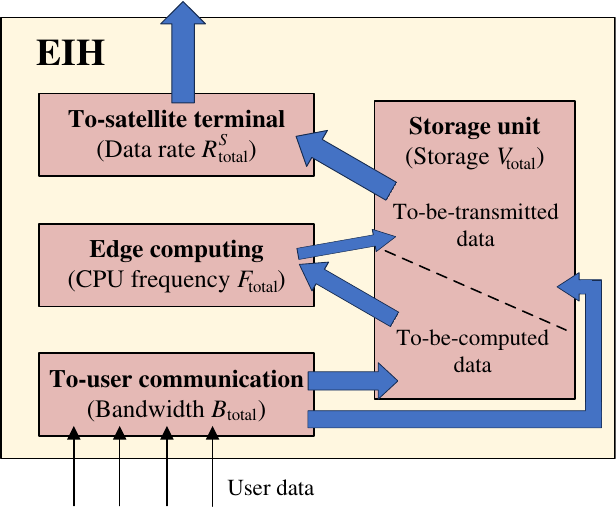}%
	\caption{Data flow diagram of the EIH during the data uploading process.}
	\label{fig data flow}
\end{figure}

For the user-EIH transmissions, the received signal at the EIH from user $u$ is given by
\begin{equation}
	y_u = g_u x_u + n_u,
\end{equation}
where $x_u$ denotes the transmit signal, $n_u$ denotes the additive white Gaussian noise, and $g_u$ denotes the user-EIH channel gain, which can be modeled as
\begin{equation}
	g_u = s_u \cdot l_u,
\end{equation}
where $s_u$ represents the fast-varying small-scale fading, following a complex Gaussian distribution with zero mean and unit variance (\textit{i.e.}, Rayleigh fading), and $l_u$ is the large-scale channel which can be expressed as \cite{channel}
\begin{equation}
	l_u = \frac{c}{4\pi f d_u} 10^{ -\frac{1}{20} \left( \frac{\eta_{\rm LoS}-\eta_{\rm NLoS}}{1+a\exp(-b(\theta_u-a))}+\eta_{\rm NLoS} \right) }, \label{large-scale channel}
\end{equation}
where $c$ denotes the speed of light, $f$ denotes the carrier frequency, and $(\eta_{\rm LoS},\eta_{\rm NLoS},a,b)$ are constants related to the propagation environment. Besides, $d_u$ and $\theta_u$ are the distance and elevation angle between the ground user and the UAV, respectively, which are given by
\begin{equation}
	d_u = \sqrt{(h_x-h_{u,x})^2+(h_y-h_{u,y})^2+H^2}, \label{distance}
\end{equation}
\begin{equation}
	\theta_u = \tan^{-1}\left(\frac{H}{\sqrt{(h_x-h_{u,x})^2+(h_y-h_{u,y})^2}}\right). \label{elevation angle}
\end{equation}
Since sensors' coordinates and the EIH's height are fixed, $l_u$ is regarded as a function of $(h_x,h_y)$. The ergodic transmission rate for user $u$ could be formulated as
\begin{equation}
	R_u = B_u \mathbb{E} \left[ \log_2\left(1+\frac{p_u|g_u|^2}{\sigma^2}\right) \right],
	\label{ergodic rate}
\end{equation}
where $p_u = \mathbb{E} |x_u|^2$ denotes the transmit power and $\sigma^2$ denotes the noise power.

Since a data-flow-based model is considered in the EIH, for user $u$, the overall upload latency $T_u$ and minimum required storage $V_u$ are subject to the data rates of the corresponding data flows. Furthermore, these data flow rates could be determined by user $u$'s allocated communication resources ($R_u,R_u^S$) and computing resources ($F_u$), as well as data scheduling ($\eta_u$). Referring to \cite{EIH RW sat-UAV MEC 05}, the piecewise expressions of $T_u(R_u,R_u^S,F_u,\eta_u)$ and $V_u(R_u,R_u^S,F_u,\eta_u)$ are given in Table \ref{T_u V_u}. The detailed derivation of these expressions could be referred to Appendix A in \cite{EIH RW sat-UAV MEC 05}.

\begin{table*}[t]
	\caption{The expressions of user $u$'s overall upload latency $T_u(B_u,R_u^S,F_u,\eta_u)$ and minimum required storage $V_u(B_u,R_u^S,F_u,\eta_u)$ \cite{EIH RW sat-UAV MEC 05}.}
	\centering
	\renewcommand{\arraystretch}{2.1}
	\renewcommand{\cellset}{}
	\begin{tabular}{|m{1.65in}|m{2.45in}|m{1.0in}|m{1.3in}|}
		\hline
		\textbf{Condition} & \textbf{Physical interpretation} & $T_u(R_u,R_u^S,F_u,\eta_u)$ & $V_u(R_u,R_u^S,F_u,\eta_u)$ \\ \hline
		\makecell*[l]{$\eta_u R_u \geq \frac{F_u}{\rho_u}$, $\frac{F_u}{\rho_u} \geq \frac{\eta_u R_u^S}{\zeta_u\eta_u+1-\eta_u}$, \\ $\frac{\zeta_u F_u}{\rho_u}+(1-\eta_u) R_u \geq R_u^S$} & Both the \textbf{to-be-computed} and \textbf{to-be-uploaded} data accumulate in the storage during the user-EIH transmission, and the \textbf{former} runs out first afterwards. & $\frac{D_u(\zeta_u\eta_u+1-\eta_u)}{R_u^S}$ & $\frac{D_u\left[R_u - R_u^S - (1-\zeta_u) \frac{F_u}{\rho_u}\right]}{R_u}$ \\ \hline
		\makecell*[l]{$\eta_u R_u \geq \frac{F_u}{\rho_u}$, $\frac{F_u}{\rho_u} < \frac{\eta_u R_u^S}{\zeta_u\eta_u+1-\eta_u}$, \\ $\frac{\zeta_u F_u}{\rho_u}+(1-\eta_u) R_u \geq R_u^S$} & Both the \textbf{to-be-computed} and \textbf{to-be-uploaded} data accumulate in the storage during the user-EIH transmission, and the \textbf{latter} runs out first afterwards. & $\frac{\eta_u D_u \rho_u}{F_u}$ & $\frac{D_u\left[R_u - R_u^S - (1-\zeta_u) \frac{F_u}{\rho_u}\right]}{R_u}$ \\ \hline
		\makecell*[l]{$\eta_u R_u \geq \frac{F_u}{\rho_u}$, \\ $\frac{\zeta_u F_u}{\rho_u}+(1-\eta_u) R_u < R_u^S$} & Only the \textbf{to-be-computed} data accumulate in the storage during the user-EIH transmission. & $\frac{\eta_u D_u \rho_u}{F_u}$ & $\frac{D_u\left(\eta_u R_u - \frac{F_u}{\rho_u}\right)}{R_u}$ \\ \hline
		\makecell*[l]{$\eta_u R_u < \frac{F_u}{\rho_u}$, \\ $(\zeta_u\eta_u+1-\eta_u) R_u \geq R_u^S$} & Only the \textbf{to-be-uploaded} data accumulate in the storage during the user-EIH transmission. & $\frac{D_u(\zeta_u\eta_u+1-\eta_u)}{R_u^S}$ & $\frac{D_u\left[(\zeta_u\eta_u+1-\eta_u) R_u - R_u^S\right]}{R_u}$ \\ \hline
		\makecell*[l]{$\eta_u R_u < \frac{F_u}{\rho_u}$, \\ $(\zeta_u\eta_u+1-\eta_u) R_u < R_u^S$} & Neither the \textbf{to-be-computed} data nor the \textbf{to-be-uploaded} data accumulate in the storage during the user-EIH transmission. & $\frac{D_u}{R_u}$ & $0$ \\ \hline
	\end{tabular}
	\label{T_u V_u}
\end{table*}

\section{Problem Formulation and Proposed Solution}
\subsection{Problem Formulation}
The purpose of this paper is to determine the optimal location of the EIH $(h_x,h_y)$ and the optimal configuration of the EIH's communication and computing resource, including the total user-EIH transmission bandwidth $B_{\rm total}$, the total EIH-satellite transmission data rate $R^S_{\rm total}$, the EIH's total computing resource $F_{\rm total}$ and total storage resource $V_{\rm total}$. The aim is to minimize the overall resource costs while satisfying the data uploading delay constraint. 

We formulate the optimization problem as follows:
\begin{subequations} \label{problem 1}
	\begin{align}
		\mbox{(P1)} \min_{\substack{h_x, h_y, \\ B_{\rm total}, R^S_{\rm total}, \\ F_{\rm total}, V_{\rm total}, \\ \mathbf{B}, \mathbf{R}^S, \mathbf{F}, \bm{\eta}}} & a_1 B_{\rm total} + a_2 R^S_{\rm total} + a_3 F_{\rm total} + a_4 V_{\rm total} \label{problem 1a} \\
		\mathrm{s.t.} \hspace{8mm} & T_u(R_u,R_u^S,F_u,\eta_u) \leq T_{\rm req},\  u \in \mathcal{U} \label{problem 1b} \\
		& \sum_{u=1}^{U} V_u(R_u,R_u^S,F_u,\eta_u) \leq V_{\mathrm{total}}, \label{problem 1c} \\
		& \sum_{u=1}^{U} B_u \leq B_{\mathrm{total}}, \label{problem 1d} \\
		& \sum_{u=1}^{U} R_u^S \leq R^S_{\mathrm{total}}, \label{problem 1e} \\
		& \sum_{u=1}^{U} F_u \leq F_{\mathrm{total}}, \label{problem 1f} \\
		& B_u, R_u^S, F_u \geq 0, \  u \in \mathcal{U}, \label{problem 1g} \\
		& 0 \leq \eta_u \leq 1, \  u \in \mathcal{U}, \label{problem 1h} \\
		& B_{\rm total}, R^S_{\rm total}, F_{\rm total}, V_{\rm total} \geq 0, \label{problem 1i}
	\end{align}
\end{subequations}
where the overall resource cost is denoted by the weighted sum of the configured resources, and $(a_1,a_2,a_3,a_4)$ are the weight parameters. Since the resource configuration design is coupled with the resource orchestration design, user $u$'s allocated resources $B_u$, $R_u^S$ and $F_u$ as well as data scheduling variable $\eta_u$ are jointly optimized. We note that $R_u = R_u(B_u,h_x,h_y)$ in this problem referring to (\ref{ergodic rate}).

\subsection{Optimal Configuration of Communication and Computing Resources}
In this subsection, we first assume the EIH-enabled UAV's location $(h_x,h_y)$ is given and fixed, and we optimize the resource configuration and resource orchestration. In this case, $R_u = R_u(B_u)$ and therefore $T_u$ and $V_u$ given in Table \ref{T_u V_u} could easily be recast as $T_u(B_u,R_u^S,F_u,\eta_u)$ and $V_u(B_u,R_u^S,F_u,\eta_u)$.

We first analyze the properties of the piecewise functions $T_u(B_u,R_u^S,F_u,\eta_u)$ and $V_u(B_u,R_u^S,F_u,\eta_u)$. Specifically, we propose the following theorem:
\begin{theorem}
	Assume that the values of variables $B_u$, $R_u^S$ and $F_u$ are given, there exists an optimal value of variable $\eta_u^{\mathrm{opt}} = \eta_u^{\mathrm{opt}}(B_u,R_u^S,F_u)$ so that inequalities
	\begin{align}
		& \nonumber T_u^{\eta-\mathrm{opt}}(B_u,R_u^S,F_u) = T_u(B_u,R_u^S,F_u,\eta_u^{\mathrm{opt}}) \\
		& \hspace{25mm} \leq T_u(B_u,R_u^S,F_u,\eta_u),\   \forall 0 \leq \eta_u \leq 1, \label{eta_opt T ineq}
	\end{align}
	and
	\begin{align}
		& \nonumber V_u^{\eta-\mathrm{opt}}(B_u,R_u^S,F_u) = V_u(B_u,R_u^S,F_u,\eta_u^{\mathrm{opt}}) \\
		& \hspace{25mm} \leq V_u(B_u,R_u^S,F_u,\eta_u),\   \forall 0 \leq \eta_u \leq 1, \label{eta_opt V ineq}
	\end{align}
	hold, where the expressions of $\eta_u^{\mathrm{opt}}$, $T_u^{\eta-\mathrm{opt}}(B_u,R_u^S,F_u)$ and $V_u^{\eta-\mathrm{opt}}(B_u,R_u^S,F_u)$ are given in Table \ref{T_u V_u eta_opt}.
	\label{theo 0}
\end{theorem}
\begin{proof}[Proof:\nopunct]
	See Appendix A.
\end{proof}
We note that a proof of Theorem \ref{theo 0} based on mathematical derivations could be referred to the Appendix B of \cite{EIH RW sat-UAV MEC 05}, whereas in Appendix A we provide a more intuitive proof of the theorem based on diagrams.

By substituting $\eta_u$ in problem (P1) with $\eta_u^{\mathrm{opt}}$, we could formulate problem (P2) as:
\begin{subequations} \label{problem 2}
	\begin{align}
		\mbox{(P2)} \min_{\substack{B_{\rm total}, R^S_{\rm total}, \\ F_{\rm total}, V_{\rm total}, \\ \mathbf{B}, \mathbf{R}^S, \mathbf{F}}} & a_1 B_{\rm total} + a_2 R^S_{\rm total} + a_3 F_{\rm total} + a_4 V_{\rm total} \label{problem 2a} \\
		\mathrm{s.t.} \hspace{8mm} & T_u^{\eta-{\rm opt}}(B_u,R_u^S,F_u) \leq T_{\rm req}, \ u \in \mathcal{U}, \label{problem 2b} \\
		& \sum_{u=1}^{U} V_u^{\eta-{\rm opt}}(B_u,R_u^S,F_u) \leq V_{\mathrm{total}}, \label{problem 2c} \\
		\nonumber & \mathrm{(7d), (7e), (7f), (7g), (7i)}.
	\end{align}
\end{subequations}
The relationship between problem (P1) and problem (P2) are given in the following theorem:
% The equivalence between problem (\ref{problem 1}) and problem (\ref{problem 2}) could be easily verified based on theorem \ref{theorem eta_opt}. Specifically, for any problem (\ref{problem 1})'s feasible solution $(B^0_{\rm total}, R^{S,0}_{\rm total}, F^0_{\rm total}, V^0_{\rm total}, \mathbf{B}^0, \mathbf{R}^{S,0}, \mathbf{F}^0, \bm{\eta}^0)$, we could verify that $(B^0_{\rm total}, R^{S,0}_{\rm total}, F^0_{\rm total}, V^0_{\rm total}, \mathbf{B}^0, \mathbf{R}^{S,0}, \mathbf{F}^0, \bm{\eta}^{0,{\rm opt}})$ is also a feasible solution and has the same objective function value. Therefore, replacing the $\bm{\eta}$ in problem (\ref{problem 1}) with $\bm{\eta}^{\rm opt}$ is equivalent problem transformation.

\begin{table*}[t]
	\caption{Optimal data scheduling variable and corresponding optimized function values \cite{EIH RW sat-UAV MEC 05}.}
	\centering
	\renewcommand{\arraystretch}{3.3}
	\begin{tabular}{|m{2.3in}|m{1.05in}|m{1.05in}|m{2.00in}|}
		\hline
		\textbf{Condition} & $\eta_u^{\mathrm{opt}}(B_u,R_u^S,F_u)$ & $T_u^{\eta-\mathrm{opt}}(B_u,R_u^S,F_u)$ & $V_u^{\eta-\mathrm{opt}}(B_u,R_u^S,F_u)$ \\ \hline
		$R_u(B_u) < R_u^S$ & $0$ & $\frac{D_u}{R_u(B_u)}$ & $0$ \\ \hline
		$R_u^S \leq R_u(B_u) < \frac{R_u^S}{\zeta_u}$, $\frac{F_u}{\rho_u} < \frac{R_u(B_u) - R_u^S}{1-\zeta_u}$  & $\frac{F_u}{F_u(1-\zeta_u)+\rho_u R_u^S}$ & $\frac{D_u \rho_u}{F_u(1-\zeta_u)+\rho_u R_u^S}$ & $\frac{D_u}{R_u(B_u)}\left[R_u(B_u) - R_u^S - (1-\zeta_u) \frac{F_u}{\rho_u}\right]$ \\ \hline
		$R_u^S \leq R_u(B_u) < \frac{R_u^S}{\zeta_u}$, $\frac{F_u}{\rho_u} \geq \frac{R_u(B_u) - R_u^S}{1-\zeta_u}$ & $\frac{R_u(B_u) - R_u^S}{(1-\zeta_u) R_u(B_u)}$ & $\frac{D_u}{R_u(B_u)}$ & $0$ \\ \hline
		$R_u(B_u) \geq \frac{R_u^S}{\zeta_u}$, $\frac{F_u}{\rho_u} < \frac{R_u^S}{\zeta_u}$ & $\frac{F_u}{F_u(1-\zeta_u)+\rho_u R_u^S}$ & $\frac{D_u \rho_u}{F_u(1-\zeta_u)+\rho_u R_u^S}$ & $\frac{D_u}{R_u(B_u)}\left[R_u(B_u) - R_u^S - (1-\zeta_u) \frac{F_u}{\rho_u}\right]$ \\ \hline
		$R_u(B_u) \geq \frac{R_u^S}{\zeta_u}$, $\frac{R_u^S}{\zeta_u} \leq \frac{F_u}{\rho_u} < R_u(B_u)$ & $1$ & $\frac{\zeta_u D_u}{R_u^S}$ & $\frac{D_u}{R_u(B_u)}\left[R_u(B_u) - R_u^S - (1-\zeta_u) \frac{F_u}{\rho_u}\right]$ \\ \hline
		$R_u(B_u) \geq \frac{R_u^S}{\zeta_u}$, $\frac{F_u}{\rho_u} \geq R_u(B_u)$ & $1$ & $\frac{\zeta_u D_u}{R_u^S}$ & $\frac{D_u}{R_u(B_u)}(\zeta_u R_u(B_u) - R_u^S)$ \\ \hline
	\end{tabular}
	\label{T_u V_u eta_opt}
\end{table*}

\begin{theorem}
	If
	\begin{equation}
		\nonumber \mathbf{x}_{\rm (P2)}^* = (B_{\rm total}^*, R^{S,*}_{\rm total}, F_{\rm total}^*, V_{\rm total}^*, \mathbf{B}^*, \mathbf{R}^{S,*}, \mathbf{F}^*)
	\end{equation}
	is an optimal solution to problem (P2), then
	\begin{align}
		\nonumber \mathbf{x}_{\rm (P1)}^* = (B_{\rm total}^*, R^{S,*}_{\rm total}, F_{\rm total}^*, & V_{\rm total}^*, \mathbf{B}^*, \mathbf{R}^{S,*}, \mathbf{F}^*, \\
		\nonumber & \{\eta_u^{\mathrm{opt}}(B_u^*,R_u^{S,*},F_u^*)\}_{u \in \mathcal{U}})
	\end{align}
	is an optimal solution to (P1).
	\label{theo 1}
\end{theorem}
\begin{proof}[Proof:\nopunct]
	We adopt the method of proof by contradiction. Assume that $\mathbf{x}_{\rm (P2)}^*$ is an optimal solution while $\mathbf{x}_{\rm (P1)}^*$ is not optimal. This means that some feasible solution to problem (P1), presumably denoted as $\mathbf{x}_{\rm (P1)}^{(1)} = (B_{\rm total}^{(1)}, R^{S,(1)}_{\rm total}, F_{\rm total}^{(1)},\\ V_{\rm total}^{(1)}, \mathbf{B}^{(1)}, \mathbf{R}^{S,{(1)}}, \mathbf{F}^{(1)}, \bm{\eta}^{(1)})$, could achieve lower objective function value compared with $\mathbf{x}_{\rm (P1)}^*$. Specifically, we have
	\begin{align}
		\nonumber & a_1 B_{\rm total}^{(1)} + a_2 R^{S,(1)}_{\rm total} + a_3 F_{\rm total}^{(1)} + a_4 V_{\rm total}^{(1)} \\
		& \hspace{10mm} < a_1 B_{\rm total}^* + a_2 R^{S,*}_{\rm total} + a_3 F_{\rm total}^* + a_4 V_{\rm total}^*. \label{theo 1 ineq}
	\end{align}
	
	Based on (\ref{eta_opt T ineq}) and (\ref{eta_opt V ineq}), $\mathbf{x}_{\rm (P1)}^{(2)} = (B_{\rm total}^{(1)}, R^{S,(1)}_{\rm total}, F_{\rm total}^{(1)}, V_{\rm total}^{(1)},\\ \mathbf{B}^{(1)}, \mathbf{R}^{S,{(1)}}, \mathbf{F}^{(1)}, \{\eta_u^{\mathrm{opt}}(B_u^{(1)},R_u^{S,(1)},F_u^{(1)})\}_{u \in \mathcal{U}})$ is also a feasible solution to (P1) and achieves the same objective function value as $\mathbf{x}_{\rm (P1)}^{(1)}$. Therefore, $\mathbf{x}_{\rm (P2)}^{(1)} = (B_{\rm total}^{(1)}, R^{S,(1)}_{\rm total},\\ F_{\rm total}^{(1)}, V_{\rm total}^{(1)}, \mathbf{B}^{(1)}, \mathbf{R}^{S,{(1)}}, \mathbf{F}^{(1)})$ is a feasible solution to (P2). Since (\ref{theo 1 ineq}) holds, $\mathbf{x}_{\rm (P2)}^{(1)}$ could achieve lower objective function value compared with $\mathbf{x}_{\rm (P2)}^*$. This contradicts the assumption of $\mathbf{x}_{\rm (P2)}^*$ being optimal, which proves the theorem.
\end{proof}

Theorem \ref{theo 1} suggests that we could obtain the optimal solution to (P1) from the optimal solution to (P2). However, (P2) is still difficult to solve since $T_u^{\eta-{\rm opt}}(B_u,R_u^S,F_u)$ and $V_u^{\eta-{\rm opt}}(B_u,R_u^S,F_u)$ are also complex piecewise functions.

To resolve this issue, we formulate the following problem by adding constraints (\ref{problem 3g}) and (\ref{problem 3h}) to (P2) as:
\begin{subequations} \label{problem 3}
	\begin{align}
		\mbox{(P3)} \min_{\substack{B_{\rm total}, R^S_{\rm total}, \\ F_{\rm total}, V_{\rm total}, \\ \mathbf{B}, \mathbf{R}^S, \mathbf{F}}} & a_1 B_{\rm total} + a_2 R^S_{\rm total} + a_3 F_{\rm total} + a_4 V_{\rm total} \label{problem 3a} \\
		\mathrm{s.t.} \hspace{5mm} & R_u^S \leq R_u(B_u) \leq \frac{R_u^S}{\zeta_u}, \ u \in \mathcal{U}, \label{problem 3g} \\
		& F_u = \rho_u \frac{R_u(B_u)-R_u^S}{1-\zeta_u}, \ u \in \mathcal{U}, \label{problem 3h} \\
		\nonumber & \mathrm{(10b), (10c), (7d), (7e), (7f), (7g), (7i).}
	\end{align}
\end{subequations}
The relationship between problem (P2) and problem (P3) are given in the following theorem:

\begin{theorem}
	If $\mathbf{x}^* = (B_{\rm total}^*, R^{S,*}_{\rm total}, F_{\rm total}^*, V_{\rm total}^*, \mathbf{B}^*, \mathbf{R}^{S,*},\\ \mathbf{F}^*)$ is an optimal solution to problem (P3), then $\mathbf{x}^*$ is an optimal solution to (P2).
	\label{theo 2}
\end{theorem}
\begin{proof}[Proof:\nopunct]
	See Appendix B.
\end{proof}

Theorem \ref{theo 2} suggests that we could obtain an optimal solution to (P2) by deriving from an optimal solution to (P3). Moreover, since (P3) reduces the feasible region compared with (P2) by adding constraints (\ref{problem 3g})-(\ref{problem 3h}), functions $T_u^{\eta-\mathrm{opt}}$ and $V_u^{\eta-\mathrm{opt}}$ take a certain piece (row 3 of Table \ref{T_u V_u eta_opt}). Therefore, (P3) could be equivalently rewritten as:

\begin{subequations} \label{problem 4}
	\begin{align}
		\mbox{(P4)} \min_{\substack{B_{\rm total}, R^S_{\rm total}, \\ F_{\rm total}, \mathbf{B}, \mathbf{R}^S, \mathbf{F}}} & a_1 B_{\rm total} + a_2 R^S_{\rm total} + a_3 F_{\rm total} + a_4 V_{\rm total} \label{problem 4a} \\
		\mathrm{s.t.} \hspace{5mm} & \frac{D_u}{R_u(B_u)} \leq T_{\rm req}, \ u \in \mathcal{U}, \label{problem 4b} \\
		& B_{\rm total}, R^S_{\rm total}, F_{\rm total} \geq 0, \label{problem 4h} \\
		\nonumber & \mathrm{(12b), (12c), (7d), (7e), (7f), (7g).}
	\end{align}
\end{subequations}
We note that within the feasible region constrained by (\ref{problem 3g})-(\ref{problem 3h}), the function value of $V_u^{\eta-\mathrm{opt}}$ equals to 0. This suggests that in the optimal solution the communication and computing data flow would reach a balance, and no data is accumulated at the EIH. The required storage resource $V_{\rm total}$ is thus 0 in this case.

However, this problem is non-convex and therefore cannot be solved directly. Since the variables are strongly coupled in the constraints, successive convex approximation cannot be applied to obtain an approximate optimal solution.

To solve (P4), we propose the following theorem:
\begin{theorem}
	Assuming that $R_u(B_u)$ is strictly increasing with $B_u$, there exists at least one optimal solution $(B_{\rm total}^*, R^{S,*}_{\rm total}, F^*_{\rm total}, \mathbf{B}^*, \mathbf{R}^{S,*}, \mathbf{F}^*)$ to (P4) that satisfies
	\begin{align}
		& B_{\rm total}^* = \sum_{u=1}^{U} R_u^{-1}(\frac{D_u}{T_{\rm req}}), \label{B_total opt} \\
		& B_u^* = R_u^{-1}(\frac{D_u}{T_{\rm req}}),\ u \in \mathcal{U} \label{B_u opt}.
	\end{align}
	\label{theorem B opt}
\end{theorem}
\begin{proof}[Proof:\nopunct]
	See Appendix C.
\end{proof}
Based on this theorem, we could further equivalently transform (P4) into the following problem:
\begin{subequations} \label{problem 5}
	\begin{align}
		\mbox{(P5)} \min_{R^S_{\rm total}, F_{\rm total}, \mathbf{R}^S} & a_2 R^S_{\rm total} + a_3 F_{\rm total} + a_1 B_{\rm total}^* \label{problem 5a} \\
		\mathrm{s.t.} \hspace{6mm} & \sum_{u=1}^{U} R_u^S \leq R^S_{\mathrm{total}}, \label{problem 5b} \\
		& \sum_{u=1}^{U} \rho_u \frac{D_u/T_{\rm req}-R_u^S}{1-\zeta_u} \leq F_{\mathrm{total}}, \label{problem 5c} \\
		& \zeta_u \frac{D_u}{T_{\rm req}} \leq R_u^S \leq \frac{D_u}{T_{\rm req}}, \ u \in \mathcal{U}, \label{problem 5d} \\
		& R^S_{\rm total}, F_{\rm total} \geq 0. \label{problem 5e}
	\end{align}
\end{subequations}

This optimization problem is a linear programming problem. The optimal solution to this problem is given in the following theorem
\begin{theorem}
	The optimal solution to (P5) is defined as $(R^{S,*}_{\rm total}, F^*_{\rm total}, \mathbf{R}^{S,*})$, the expressions of which are given as follows.
	\begin{align}
		& R_u^{S,*} = \frac{D_u}{T_{\rm req}} \left[ 1-(1-\zeta_u) \cdot {\rm I}\left( \frac{\rho_u}{1-\zeta_u} \leq \frac{a_2}{a_3} \right) \right], \label{R^S opt} \\
		& R_{\rm total}^{S,*} = \sum_{u=1}^{U} \frac{D_u}{T_{\rm req}} \left[ 1-(1-\zeta_u) \cdot {\rm I}\left( \frac{\rho_u}{1-\zeta_u} \leq \frac{a_2}{a_3} \right) \right], \label{R^S_total opt} \\
		& F_{\rm total}^* = \sum_{u=1}^{U} \rho_u \frac{D_u}{T_{\rm req}} \cdot {\rm I}\left( \frac{\rho_u}{1-\zeta_u} \leq \frac{a_2}{a_3} \right), \label{F_total opt}
	\end{align}
	where ${\rm I}(\cdot)$ is an indicator function, which equals to 1 when the statement holds, and 0 if it does not hold.
	\label{theorem R^S F opt}
\end{theorem}
\begin{proof}[Proof:\nopunct]
	See Appendix D.
\end{proof}
We could see that the optimal configuration of to-satellite data rate $R_{\rm total}^{S,*}$ and computing frequency $F_{\rm total}^*$ is determined by a threshold $\frac{\rho_u}{1-\zeta_u} \leq \frac{a_2}{a_3}$, where $\frac{\rho_u}{1-\zeta_u}$ is the number of computing CPU cycles required for each bit of computing output. If the threshold holds, it means that the cost-effectiveness ratio of computing is relatively low, and thus sensor $u$'s data are computed before uploading to the satellite. If it does not hold, sensor $u$'s data are directly uploaded to the satellite and the required computing frequency is 0.

\subsection{Optimal location of the EIH}

In the previous subsection, we acquire the closed-formed expressions of the optimal resource configuration (\ref{B_total opt}), (\ref{R^S_total opt}) and (\ref{F_total opt}), assuming that the location of the EIH-enabled UAV $(h_x,h_y)$ is given and fixed. According to these expressions, we note that the EIH's location only affects the EIH-user communication resource configuration. Therefore, the optimization problem of the EIH location could be formulated as:
\begin{subequations} \label{problem 6}
	\begin{align}
		\mbox{(P6)} \min_{h_x,h_y,\mathbf{B}} \hspace{1mm}  & a_1 \sum_{u=1}^{U} B_u + a_2 R^{S,*}_{\rm total} + a_3 F^*_{\rm total} \label{problem 6a} \\
		\nonumber \mathrm{s.t.} \hspace{2mm} & B_u \mathbb{E} \left[ \log_2\left(1+\frac{p_u|s_u|^2 l_u(h_x,h_y)^2}{\sigma^2}\right) \right] \\
		& \hspace{35mm} = \frac{D_u}{T_{\rm req}},\ u \in \mathcal{U}, \label{problem 6b}
	\end{align}
\end{subequations}
where the expression of $l_u(h_x,h_y)$ could refer to (\ref{large-scale channel})-(\ref{elevation angle}). With the expectation operators in constraint (\ref{problem 6b}), this problem could be difficult to solve. Referring to \cite{ergodic rate}, we propose a rather accurate approximation method based on the random matrix theory
\begin{equation}
	\tilde{R}_u = \frac{B_u}{\log 2} \cdot \left[\log\left(1+\frac{p_u l_u^2}{\sigma^2 \nu_u}\right) - \frac{p_u l_u^2}{p_u l_u^2+\sigma^2 \nu_u} + \log\nu_u \right],
\end{equation}
where the introduced variable $\nu_u$ satisfies $\nu_u \geq 1$ and
\begin{equation}
	1-\frac{1}{\nu_u}=\frac{p_u l_u^2}{p_u l_u^2+\sigma^2 \nu_u}.
\end{equation}
Through this approximation method, we recast optimization problem (P6) as:
\begin{subequations} \label{problem 7}
	\begin{align}
		\mbox{(P7)} \min_{h_x,h_y,\mathbf{B},\bm{\nu}} \hspace{1mm}  & a_1 \sum_{u=1}^{U} B_u + a_2 R^{S,*}_{\rm total} + a_3 F^*_{\rm total} \label{problem 7a} \\
		\nonumber \mathrm{s.t.} \hspace{3mm} & \log\left(1+\frac{p_u l_u^2}{\sigma^2 \nu_u}\right) + \log\nu_u - \frac{p_u l_u^2}{p_u l_u^2+\sigma^2 \nu_u}\\
		& \hspace{26mm} = \frac{D_u\log2}{T_{\rm req} B_u},\ u \in \mathcal{U}, \label{problem 7b} \\
		& 1-\frac{1}{\nu_u}=\frac{p_u l_u^2}{p_u l_u^2+\sigma^2 \nu_u},\ u \in \mathcal{U}, \label{problem 7c} \\
		& \nu_u \geq 1,\ u \in \mathcal{U}. \label{problem 7d}
	\end{align}
\end{subequations}
This problem is still complicated and cannot be solved directly. By introducing slack variables $\mathbf{Q}=\{Q_u, u \in \mathcal{U}\}$, $\mathbf{Pr}=\{Pr_u, u \in \mathcal{U}\}$ and $\bm{\theta}=\{\theta_u, u \in \mathcal{U}\}$, a new optimization is proposed as:
\begin{subequations} \label{problem 8}
	\begin{align}
		\mbox{(P8)} \min_{\substack{h_x, h_y, \\ \mathbf{B}, \bm{\nu}, \\ \mathbf{Q}, \mathbf{Pr}, \bm{\theta}}} \hspace{1mm}  & a_1 \sum_{u=1}^{U} B_u + a_2 R^{S,*}_{\rm total} + a_3 F^*_{\rm total} \label{problem 8a} \\
		\mathrm{s.t.} \hspace{3mm} & 2\log\nu_u+\frac{1}{\nu_u}-1 \geq \frac{D_u\log2}{T_{\rm req} B_u},\ u \in \mathcal{U}, \label{problem 8b} \\
		\nonumber & \log[\nu_u(\nu_u-1)] \leq \frac{\log 10}{10} (\eta_{\rm NLoS}-\eta_{\rm LoS}) Pr_u \\
		& \hspace{27mm} -\log Q_u + C,\ u \in \mathcal{U}, \label{problem 8c} \\
		\nonumber & Q_u \geq (h_x-h_{u,x})^2+(h_y-h_{u,y})^2 \\
		& \hspace{38mm} +H^2 ,\ u \in \mathcal{U}, \label{problem 8d} \\
		& \frac{1}{Pr_u} \geq 1+a \cdot {\rm e}^{-b(\theta_u-a)},\ u \in \mathcal{U}, \label{problem 8e} \\
		& \sin \theta_u \leq \frac{H}{\sqrt{Q_u}},\ u \in \mathcal{U}, \label{problem 8f} \\
		& \nu_u \geq 1,\ u \in \mathcal{U}, \label{problem 8g} \\
		& 0 \leq Pr_u \leq 1,\ u \in \mathcal{U}, \label{problem 8h} \\
		& 0 \leq \theta_u \leq \pi/2,\ u \in \mathcal{U}, \label{problem 8i}
	\end{align}
\end{subequations}
where $C = \log(p_u/\sigma^2)+2\log(c/4\pi f)-(\log 10/10) \cdot \eta_{\rm NLoS}$. The relationship between problems (P7) and (P8) is given in the following theorem
\begin{theorem}
	Problems (P7) and (P8) are equivalent.
	\label{theo placement prob equiv}
\end{theorem}
\begin{proof}[Proof:\nopunct]
	See Appendix E.
\end{proof}

Problem (P8) is still non-convex. To solve the problem, we apply the successive convex approximation (SCA) method and replace the non-convex terms in the constraints with the corresponding first-order Taylor expansions. Specifically, (\ref{problem 8c}) is replaced by
\begin{align}
	\nonumber & \log[\nu_u'(\nu_u'-1)] + \frac{2\nu_u'-1}{\nu_u'(\nu_u'-1)}(\nu_u-\nu_u') \\
	\nonumber & \hspace{15mm} \leq \frac{\log 10}{10} (\eta_{\rm NLoS}-\eta_{\rm LoS}) Pr_u \\
	& \hspace{20mm} -\left(\log Q_u'+\frac{Q_u-Q_u'}{Q_u'}\right) + C,\ u \in \mathcal{U}, \label{problem 8c iter}
\end{align}
(\ref{problem 8e}) is replaced by
\begin{equation}
	\frac{1}{Pr_u'}-\frac{Pr_u-Pr_u'}{(Pr_u')^2} \geq 1+a \cdot {\rm e}^{-b(\theta_u-a)},\ u \in \mathcal{U}, \label{problem 8e iter}
\end{equation}
and (\ref{problem 8f}) is replaced by
\begin{equation}
	\sin \theta_u'+\cos \theta_u'(\theta_u-\theta_u')  \leq \frac{H}{\sqrt{Q_u'}}-\frac{H(Q_u-Q_u')}{2 Q_u'\sqrt{Q_u'}},\ u \in \mathcal{U}. \label{problem 8f iter}
\end{equation}
The new optimization problem after these replacements is given by:
\begin{subequations} \label{problem 9}
	\begin{align}
		\mbox{(P9)} \min_{\substack{h_x, h_y, \\ \mathbf{B}, \bm{\nu}, \\ \mathbf{Q}, \mathbf{Pr}, \bm{\theta}}} \hspace{1mm}  & a_1 \sum_{u=1}^{U} B_u + a_2 R^{S,*}_{\rm total} + a_3 F^*_{\rm total} \label{problem 9a} \\
		\nonumber \mathrm{s.t.} \hspace{3mm} & \mathrm{(\ref{problem 8b}), (\ref{problem 8d}), (\ref{problem 8g})-(\ref{problem 8i}), (\ref{problem 8c iter})-(\ref{problem 8f iter})},
	\end{align}
\end{subequations}
which is a convex problem. The iterative algorithm to solve problem (P8) is summarized in Algorithm \ref{algorithm 1}.
\begin{algorithm}[H]
	\caption{Iterative EIH Location Optimization Algorithm.}\label{algorithm 1}
	\begin{algorithmic}
		\STATE 
		\STATE {\bf Initialization:}
		\STATE \hspace{0.5cm}$(h_x,h_y)$ set as the geometric center of sensors in $\mathcal{U}$
		\STATE \hspace{0.5cm}$Q_u = (h_x-h_{u,x})^2+(h_y-h_{u,y})^2+H^2$
		\STATE \hspace{0.5cm}$\theta_u = \sin^{-1}(H/\sqrt{Q_u})$
		\STATE \hspace{0.5cm}$Pr_u = 1/(1+a\cdot {\rm e}^{-b(\theta_u-a)})$
		\STATE \hspace{0.5cm}$\nu_u = \frac{1}{2}+\sqrt{\frac{1}{4}+\frac{p_u c^2}{\sigma^2 (4\pi f)^2 Q_u} 10^{\frac{\eta_{\rm NLoS}-\eta_{\rm LoS}}{10} Pr_u - \frac{\eta_{\rm NLoS}}{10}}}$
		\STATE \hspace{0.5cm}$B_u = (D_u\log2)/[T_{\rm req}(2\log\nu_u+\frac{1}{\nu_u}-1)]$
		\STATE {\bf repeat}
		\STATE \hspace{0.5cm} substitute $(\mathbf{B}',\bm{\nu}',\mathbf{Q}',\mathbf{Pr}',\bm{\theta}')$ with $(\mathbf{B},\bm{\nu},\mathbf{Q},\mathbf{Pr},\bm{\theta})$
		\STATE \hspace{0.5cm} calculate $(h_x,h_y,\mathbf{B},\bm{\nu},\mathbf{Q},\mathbf{Pr},\bm{\theta})$ via solving (P9)
		\STATE {\bf until} $|\sum_{u=1}^{U}B_u - \sum_{u=1}^{U} B_u'|/\sum_{u=1}^{U} B_u' < \epsilon$
	\end{algorithmic}
\end{algorithm}

\section{Simulation Results}
In this section, we present numerical results to evaluate the performance of our proposed configuration scheme. We set the number of users as $U = 5$. We assume the maximum sensing data size $D_{\rm max}$ ranges from 1 Mbits to 100 Mbits. User $u$' sensing data size is uniformly distributed in the range $[0.1D_{\rm max},D_{\rm max}]$. The computing intensity $\rho_u$ is uniformly distributed over $[1000,5000]$ cycles/bit, and the computing output-to-input ratio is uniformly distributed over $[0.01,0.1]$. The EIH-enabled UAV's position is set as $[0,0,1000]$ m. The positions of the users are uniformly distributed within circle $\{[h_{u,x},h_{u,y},0]|\sqrt{h_{u,x}^2+h_{u,y}^2} \leq 1000\}$ m. The large-scale channel coefficients are set as $(\eta_{\rm LoS},\eta_{\rm NLoS},a,b) = (0.1,21,4.880,0.429)$ \cite{channel coeff}. We set the carrier frequency as $f = 5.8$ GHz, and the light speed as $c = 3\times 10^8$ m/s. The transmit power of user $u$ is set to be $p_u = 1$ W, and the noise power spectral density is $\sigma^2 = -114$ dBm. The weight factors are set as $a_1 = 1$ (MHz)$^{-1}$, $a_2 = 3$ (Mbits/s)$^{-1}$, $a_3 = 1\times10^{-3}$ (Mcycles/s)$^{-1}$. The required task latency is $T_{\rm req} = 100$ s.

\begin{figure}[b]
	\centering
	\includegraphics[width=3.3in]{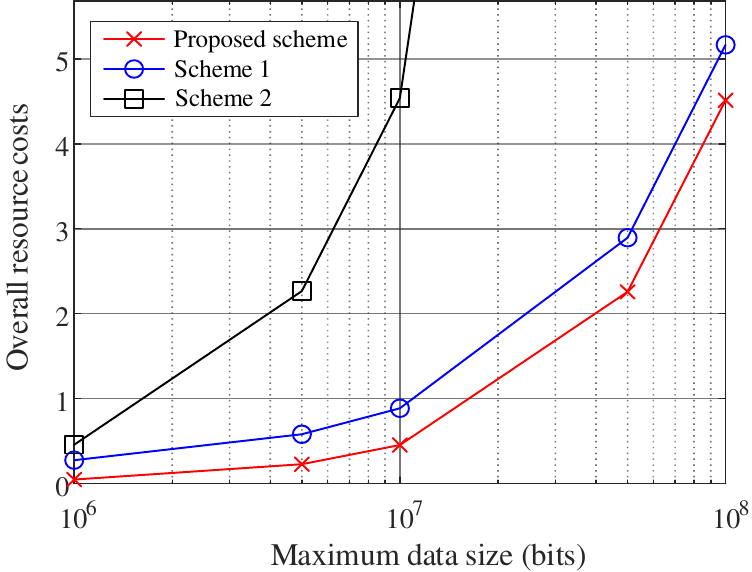}
	\caption{Relationship between the overall latency and the UAV-satellite transmission bandwidth.}
	\label{fig main_compare}
\end{figure}

As shown in Fig. \ref{fig main_compare}, we compare our proposed resource configuration scheme with two other schemes.
\begin{itemize}
	\item Scheme 1: Jointly optimizing the resource configuration and the resource orchestration, assuming that the computation and communications proceed sequentially.
	\item Scheme 2: Optimizing the resource configuration assuming that all resources are equally allocated.
\end{itemize}
Specifically, we present the overall resource costs attained by all three schemes with varying maximum data size. The results are averaged over 50 randomly generated system topologies. We can observe that the proposed scheme could lower the resource cost by about 20\% compared to Scheme 1. This is because the proposed scheme manages to optimize the resource configuration under the assumption that the computation and communications proceed concurrently, which could remarkably improve the system's latency performance and resource efficiency.

\begin{figure*}[t]
	\centering
	\includegraphics[width=7.1in]{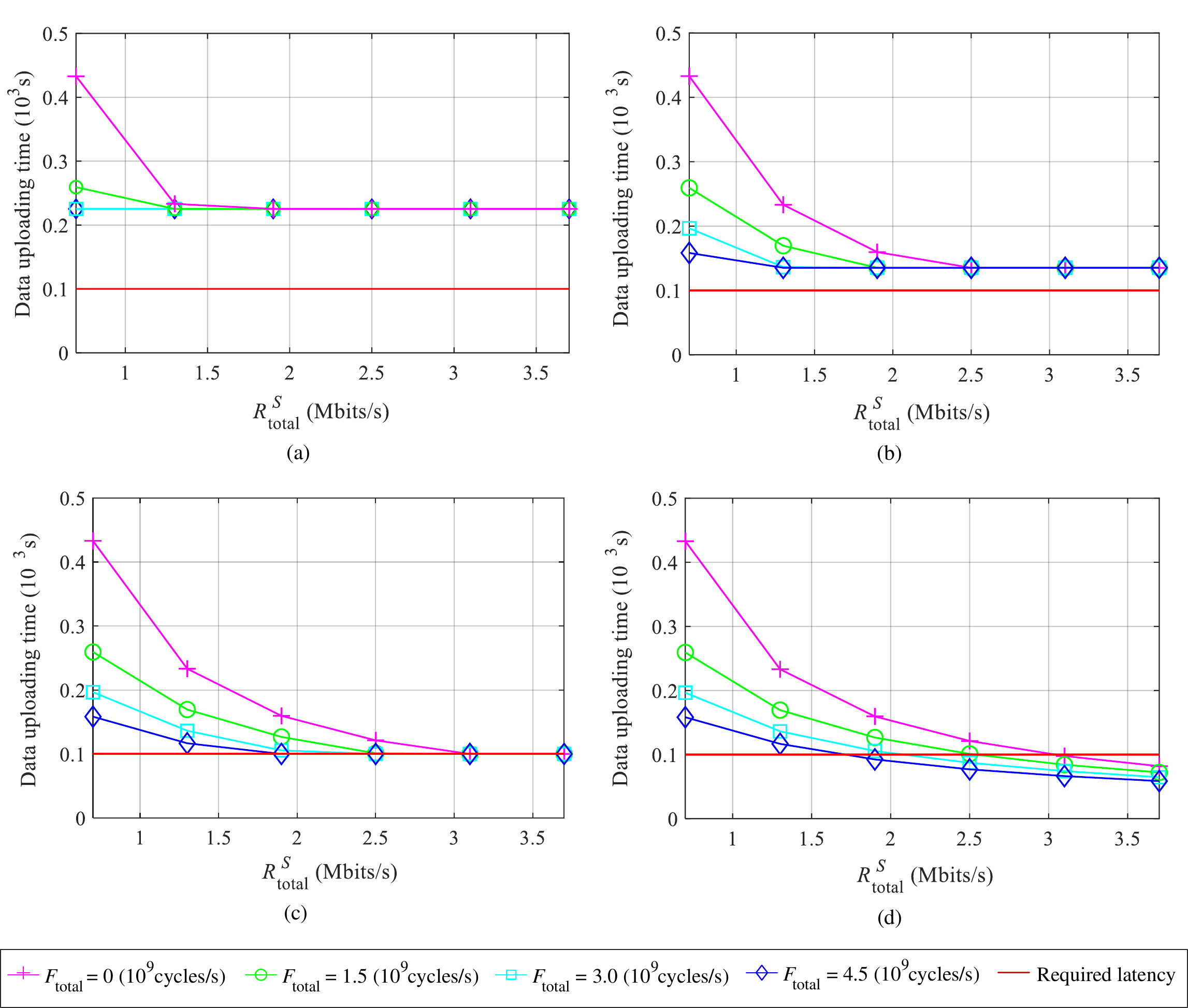}
	\caption{Relationship between the data offloading time and the configured to-satellite data rate $R_{\rm total}^S$ and CPU frequency $F_{\rm total}$, under four different to-user communication bandwidth configurations: (a) $B_{\rm total} = 0.3$ MHz; (b) $B_{\rm total} = 0.5$ MHz; (c) $B_{\rm total} = 0.676$ MHz (optimal configuration $B_{\rm total}^*$); (d) $B_{\rm total} = 1.2$ MHz.}
	\label{fig main_optBtot}
\end{figure*}

In Fig. \ref{fig main_optBtot}, we verify the optimality of the to-user communication bandwidth configuration $B_{\rm total}^*$ acquired by the proposed scheme. Specifically, we consider four different scenarios of to-user bandwidth configuration $B_{\rm total}$. We depict the relationship between the obtained data offloading time and the configured to-satellite data rate $R_{\rm total}^S$ and CPU frequency $F_{\rm total}$ for each scenario, and compare the offloading time with the latency requirement. As shown in Fig. \ref{fig main_optBtot}(a), we have $B_{\rm total}<B_{\rm total}^*$, and the obtained offloading time always exceeds the latency requirement, which means the task cannot be completed in time. In Fig. \ref{fig main_optBtot}(b), the configured bandwidth $B_{\rm total}$ increases but is still less than $B_{\rm total}^*$. The offloading latency decreases correspondingly but cannot reach the task required latency. In Fig. \ref{fig main_optBtot}(c), we have $B_{\rm total}=B_{\rm total}^*$, and the offloading time reaches exactly the latency requirement when resources $R_{\rm total}^S$ and $F_{\rm total}$ are provided adequately. In Fig. \ref{fig main_optBtot}(d), inequality $B_{\rm total}>B_{\rm total}^*$ holds. In this case, for some configurations $(R_{\rm total}^S,F_{\rm total})$ the offloading time could be lower than the latency requirement. We can learn from the above analysis that $B_{\rm total} \geq B_{\rm total}^*$ is a necessary condition for user-EIH communication bandwidth configuration in order to satisfy the task latency requirement.

\begin{figure}[t]
	\centering
	\includegraphics[width=3.5in]{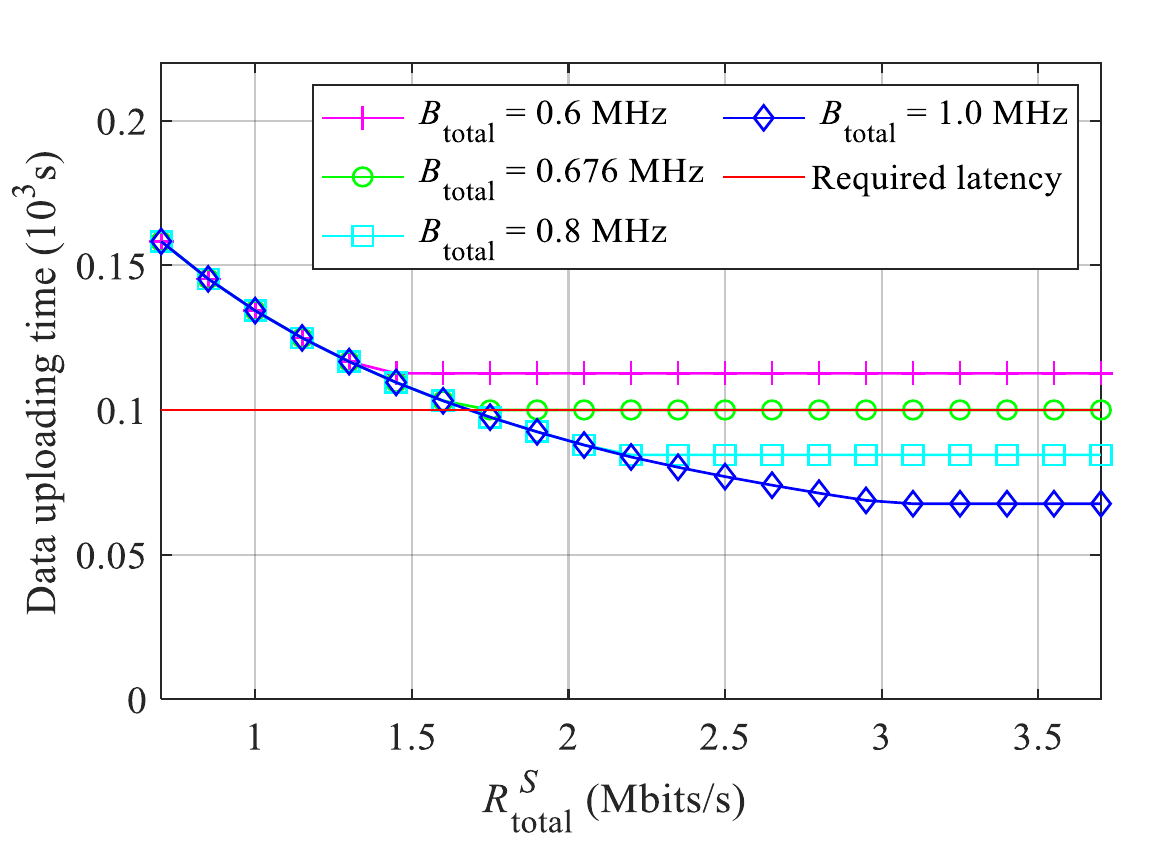}
	\caption{Relationship between the data offloading time and the configured to-satellite data rate $R_{\rm total}^S$, with given CPU frequency $F_{\rm total}$ and varying to-user communication bandwidth $B_{\rm total}$.}
	\label{fig main_optBtot2}
\end{figure}

We further verify the optimality of $B_{\rm total}^*$ in Fig. \ref{fig main_optBtot2}. Specifically, we depict the relationship between the offloading latency and the configured to-satellite data rate $R_{\rm total}^S$, with different to-user bandwidth configuration $B_{\rm total}$ (including $B_{\rm total}^* = 0.676$ MHz). The configured CPU frequency $F_{\rm total}$ is set as $4.5 \times 10^9$ cycles/s. We could observe, given a larger value of $B_{\rm total}$, the offloading latency converges to a lower level as $R_{\rm total}^S$ increases. Nevertheless, the function graphs of different $B_{\rm total}$ intersect with the constraint line at the same coordinate (if there exists an intersection). This means that if the aim is to satisfy the latency constraint, further increasing the to-user bandwidth configuration $B_{\rm total}$ from the proposed optimal value $B_{\rm total}^*$ cannot help save the to-satellite communication resources $R_{\rm total}^S$. This analysis suggests that $B_{\rm total} = B_{\rm total}^*$ is a sufficient condition for user-EIH communication bandwidth configuration. Therefore, the optimality of our proposed configuration $B_{\rm total}^*$ could be confirmed.

%However, comparing with Fig. \ref{fig main_optBtot}(b), the same configurations could provide an offloading time that equals the latency requirement under scenario $B_{\rm total}=B_{\rm total}^*$. This verifies that $B_{\rm total}^*$ is the optimal solution.

\begin{figure}[t]
	\centering
	\includegraphics[width=3.3in]{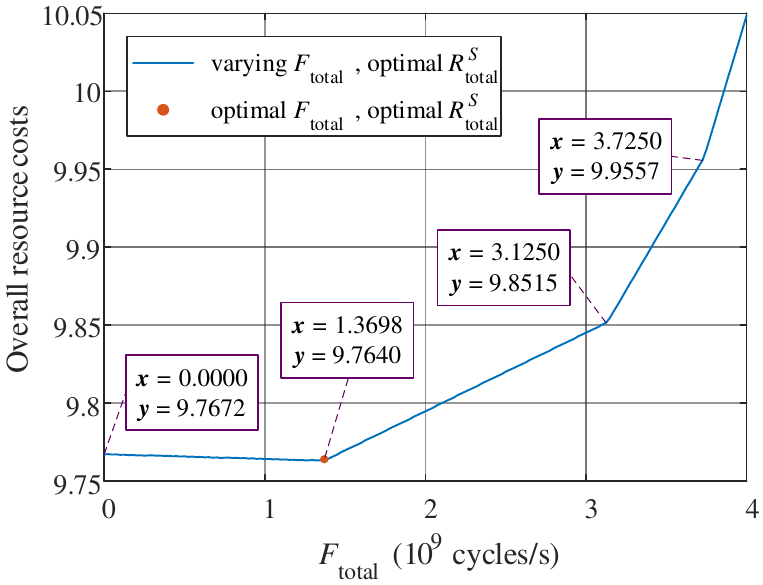}
	\caption{Relationship between the overall resource costs and the configured CPU frequency $F_{\rm total}$, considering optimal configuration of to-user communication bandwidth $B_{\rm total}^*$ and to-satellite data rate $R_{\rm total}^{S,*}$.}
	\label{fig main_optRStotFtot}
\end{figure}

\begin{figure}[t]
	\centering
	\includegraphics[width=3.3in]{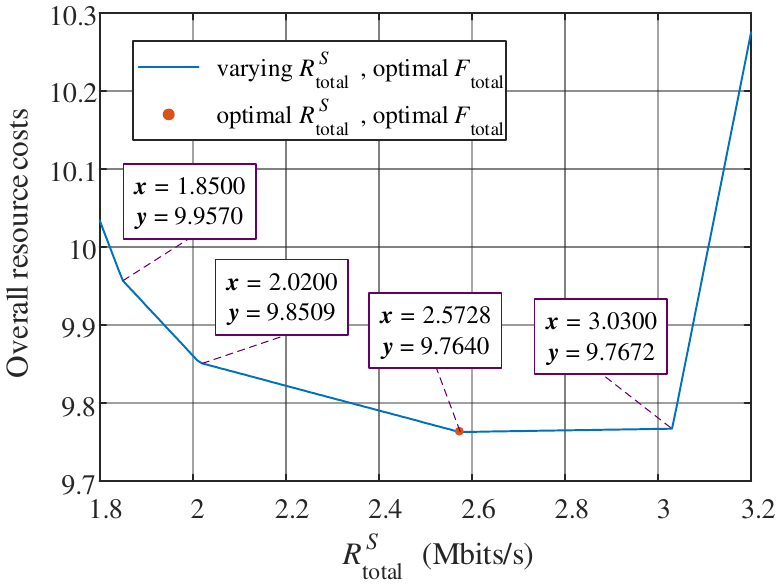}
	\caption{Relationship between the overall resource costs and the configured to-satellite data rate $R_{\rm total}^{S,*}$, considering optimal configuration of to-user communication bandwidth $B_{\rm total}^*$ and CPU frequency $F_{\rm total}$.}
	\label{fig main_optFtotRStot}
\end{figure}

Fig. \ref{fig main_optRStotFtot} and Fig. \ref{fig main_optFtotRStot} verify the optimality of the computing resource configuration $F_{\rm total}^*$ and the to-satellite data rate configuration $R_{\rm total}^{S,*}$ acquired by the proposed scheme. We first set the to-user bandwidth to $B_{\rm total}=B_{\rm total}^*$. In Fig. \ref{fig main_optRStotFtot}, we uniformly sample 200 values of $F_{\rm total}$ within range $[0,4] \times 10^9$ cycles/s. For each value, we use the binary search method to find the optimal $R_{\rm total}^{S}$ that minimizes the resource costs while meeting the latency constraint. The blue line in the figure depicts this minimized resource costs, and the red dot represents the configuration obtained by the proposed scheme. We could see that the proposed configuration reaches the minimum resource cost, which verifies the optimality of the computing resource configuration $F_{\rm total}^*$. In Fig. \ref{fig main_optFtotRStot}, with $R_{\rm total}^S$ varying from 1.8 to 3.2 Mbits/s, we also apply the binary search method to find the optimal $F_{\rm total}$. The figure shows that the proposed configuration also reaches the minimum resource cost, verifying the optimality of the to-satellite data rate configuration $R_{\rm total}^{S,*}$.

\begin{figure*}[t]
	\centering
	\includegraphics[width=7.2in]{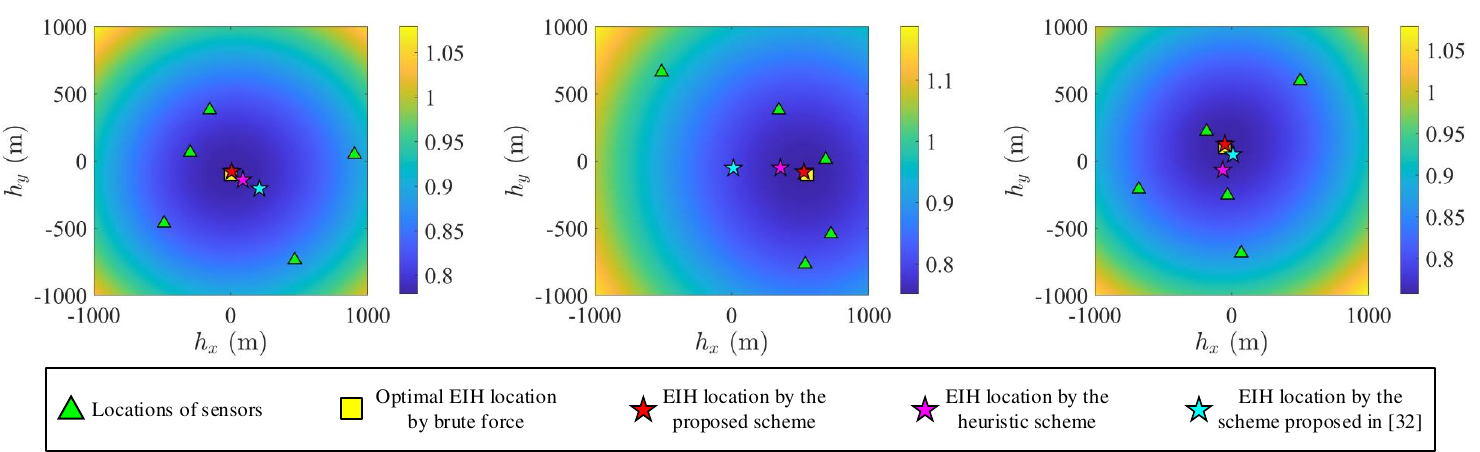}
	\caption{Relationship between the to-user communication resource costs and the position of the EIH-enabled UAV in three different topologies.}
	\label{fig main_position}
\end{figure*}

Finally, we verify the performance of our proposed location optimization scheme. As shown in Fig. \ref{fig main_position}, for three randomly-generated topologies, we first utilize the brute-force method to traverse the disaster area and obtain the to-user communication resource costs for each position. Then, we mark the EIH location by the proposed scheme and compare it with the optimal EIH location by brute force as well as the EIH locations obtained through two other schemes. One is a heuristic scheme that finds the geometric center of the sensors, and the other is the optimization scheme proposed in \cite{RW sat-UAV MEC 09}. We can see that in all topologies, the proposed scheme achieves superior performance compared with the other two schemes. This is because the proposed scheme adapts the EIH's location based on not only the sensors' locations but also their data volumes. Besides, the proposed optimal locations are basically consistent with the optimal locations based on the brute-force method, which verifies the optimality of our proposed solution.

\section{Conclusion}
This paper has investigated the resource configuration problem for edge information hub (EIH)-empowered non-terrestrial networks in post-disaster emergency scenarios. We formulated a joint optimization problem to minimize the overall resource costs while guaranteeing low-latency sensing data uploading, capturing the fundamental tradeoff between communication and computing resources. Through analytical derivations, we derived closed-form optimal solutions for resource configuration and proposed a successive convex approximation-based algorithm for location optimization. Simulation results validated the theoretical findings and demonstrated that the proposed scheme achieves approximately 20\% resource cost reduction compared with benchmark schemes. The optimality of the derived solutions was verified through comprehensive numerical analysis. This work provides systematic insights for efficient deployment of edge intelligence-enabled aerial networks in emergency response applications. Future research may extend the framework to dynamic scenarios with mobile users and multiple cooperative EIHs.

%\section*{Acknowledgments}
%This should be a simple paragraph before the References to thank those individuals and institutions who have %supported your work on this article.

\appendices

\section{Proof of Theorem \ref{theo 0}}
Theorem \ref{theo 0} proposes the optimal data scheduling variable $\eta_u^{\mathrm{opt}}$ with given values of $B_u$, $R_u^S$ and $F_u$, as well as the corresponding function values $T_u(B_u,R_u^S,F_u,\eta_u)$ and $V_u(B_u,R_u^S,F_u,\eta_u)$, which is presented in Table \ref{T_u V_u eta_opt}.

\begin{figure*}[b]
	\centering
	{\includegraphics[height=2.25in]{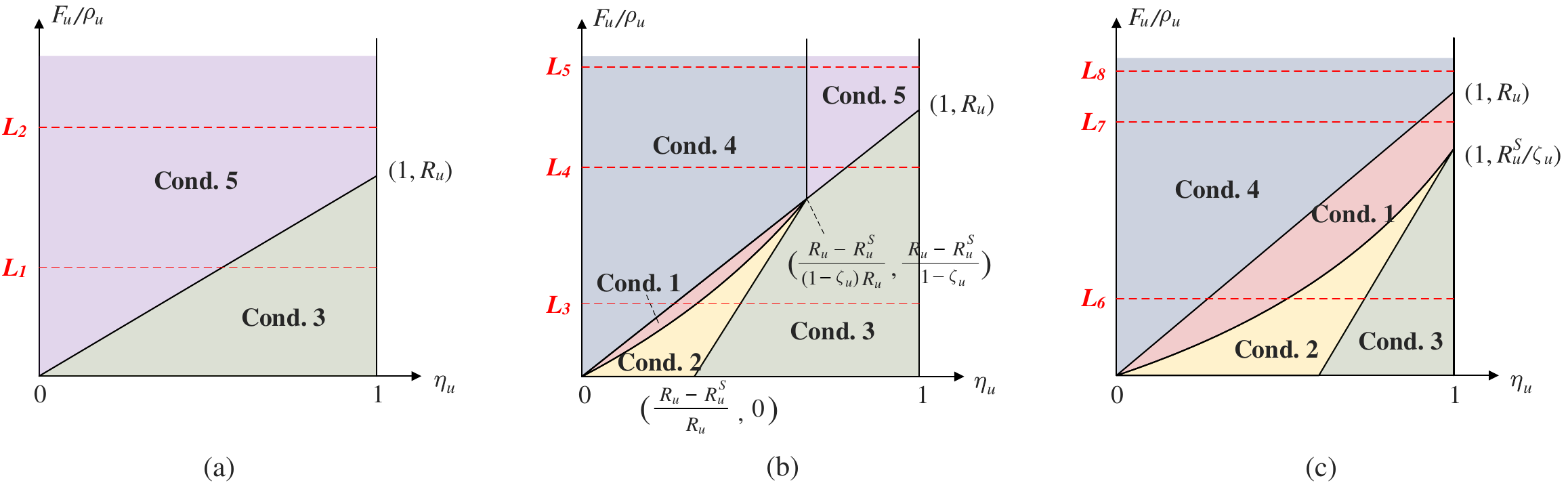}%
		\label{diag eta_opt}}
	\caption{The division of 2-D space $\{(F_u,\eta_u) |\ F_u/\rho_u \geq 0, 0 \leq \eta_u \leq 1\}$ corresponding to the conditions of Table \ref{T_u V_u} (\textit{e.g.}, Cond. 1 corresponds to the first condition of Table \ref{T_u V_u}), with different relationships between $B_u$ and $R_u^S$ values: (a) $R_u(B_u) < R_u^S$, (b) $R_u^S \leq R_u(B_u) < R_u^S/\zeta_u$,  (c) $R_u(B_u) \geq R_u^S/\zeta_u$.}
\end{figure*}

We could validate the optimality of $\eta_u^{\mathrm{opt}}$ intuitively based on Fig. 8. Specifically, we first consider three different scenarios based on the values of $B_u$ and $R_u^S$, namely $R_u(B_u) < R_u^S$, $R_u^S \leq R_u(B_u) < R_u^S/\zeta_u$, and $R_u(B_u) \geq R_u^S/\zeta_u$. For each scenario, we further divide the two-dimensional space $\{(F_u,\eta_u) |\ F_u/\rho_u \geq 0, 0 \leq \eta_u \leq 1\}$ into several zones, where each zone corresponds to a specific condition listed in Table \ref{T_u V_u}. For instance, the purple zone in Fig. 8(a) represents the four-dimensional area $\{(B_u,R_u^S,F_u,\eta_u)|\ R_u(B_u)<R_u^S, F_u/\rho_u \geq R_u(B_u) \eta_u\}$. It can be verified that this area satisfies the fifth condition in Table \ref{T_u V_u} (denoted as Cond. 5). 

Under each of the five conditions, functions $T_u$ and $V_u$ are monotonically increasing, monotonically decreasing or irrelevant with $\eta_u$. Specifically, under Cond. 1, $T_u$ is monotonically decreasing whereas $V_u$ is irrelevant with $\eta_u$. Under Cond. 2, $T_u$ is monotonically increasing whereas $V_u$ is irrelevant with $\eta_u$. Under Cond. 3, both $T_u$ and $V_u$ are monotonically increasing with $\eta_u$. Under Cond. 4, both $T_u$ and $V_u$ are monotonically decreasing with $\eta_u$. Under Cond. 5, both $T_u$ and $V_u$ are irrelevant with $\eta_u$.

In Fig. 8(a), when $F_u/\rho_u$ is relatively small (shown by line $L_1$), as $\eta_u$ increases from $0$ to $1$, $T_u$ and $V_u$ both remain unchanged within the Cond. 5 zone and both increase within the Cond. 3 zone. Therefore, the optimal $\eta_u$ could be arbitrarily selected within $[0,F_u/(\rho_u R_u)]$. When $F_u/\rho_u$ is larger (shown by line $L_2$), as $\eta_u$ increases, $T_u$ and $V_u$ both remain unchanged within the Cond. 5 zone. Therefore the optimal $\eta_u$ could be arbitrarily selected within $[0,1]$. For simplicity of expression, we select $0$ to be the optimal scheduling variable for both scenarios, and therefore
\begin{equation}
	\eta_u^{\mathrm{opt}} = 0,\hspace{20pt} R_u(B_u) < R_u^S.
\end{equation}
We substitute this $\eta_u^{\mathrm{opt}}$ into the function expressions of $T_u(B_u,R_u^S,F_u,\eta_u)$ and $V_u(B_u,R_u^S,F_u,\eta_u)$, referring to the fifth row of Table \ref{T_u V_u}, and the results are $T_u = \frac{D_u}{R_u(B_u)}$ and $V_u = 0$. This verifies the first row of Table \ref{T_u V_u eta_opt}.

In Fig. 8(b), when $F_u/\rho_u$ is relatively small (shown by line $L_3$), as $\eta_u$ increases from $0$ to $1$, it goes through zones Cond. 4, Cond. 1, Cond. 2, and Cond. 3 sequentially. During this process, $T_u$ decrease, decrease, increase and increase, whereas $V_u$ decrease, remain unchanged, remain unchanged and increase. Therefore, the optimal $\eta_u$ should be
\begin{align}
	\nonumber \eta_u^{\mathrm{opt}} & = \frac{F_u}{F_u(1-\zeta_u)+\rho_u R_u^S}, \\
	& R_u^S \leq R_u(B_u) < \frac{R_u^S}{\zeta_u},\ \frac{F_u}{\rho_u} < \frac{R_u(B_u) - R_u^S}{1-\zeta_u}.
\end{align}
We substitute this $\eta_u^{\mathrm{opt}}$ into the function expressions of $T_u(B_u,R_u^S,F_u,\eta_u)$ and $V_u(B_u,R_u^S,F_u,\eta_u)$, referring to the first (or second) row of Table \ref{T_u V_u}, and the results are $T_u = \frac{D_u \rho_u}{F_u(1-\zeta_u)+\rho_u R_u^S}$ and $V_u = \frac{D_u}{R_u(B_u)}\left[R_u(B_u) - R_u^S - (1-\zeta_u) \frac{F_u}{\rho_u}\right]$. This verifies the second row of Table \ref{T_u V_u eta_opt}.

When $F_u/\rho_u$ is larger (shown by line $L_4$), as $\eta_u$ increases from $0$ to $1$, it goes through zones Cond. 4, Cond. 5, Cond. 3 in sequence. During this process, both $T_u$ and $V_u$ first decrease, then remain unchanged and then increase. The optimal $\eta_u$ could be arbitrarily selected within $[\frac{R_u(B_u) - R_u^S}{(1-\zeta_u) R_u(B_u)},\frac{F_u}{\rho_u R_u}]$. When $F_u/\rho_u$ is further increased (shown by line $L_5$), as $\eta_u$ increases from $0$ to $1$, it goes through zones Cond. 4 and Cond. 5 in sequence. During this process, both $T_u$ and $V_u$ first decrease, then remain unchanged. The optimal $\eta_u$ could be arbitrarily selected within $[\frac{R_u(B_u) - R_u^S}{(1-\zeta_u) R_u(B_u)},1]$. For simplicity of expression, we select the optimal scheduling variable for both scenarios to be
\begin{align}
	\nonumber \eta_u^{\mathrm{opt}} & = \frac{R_u(B_u) - R_u^S}{(1-\zeta_u) R_u(B_u)}, \\
	& R_u^S \leq R_u(B_u) < \frac{R_u^S}{\zeta_u},\ \frac{F_u}{\rho_u} \geq \frac{R_u(B_u) - R_u^S}{1-\zeta_u}.
\end{align}
We substitute this $\eta_u^{\mathrm{opt}}$ into the function expressions of $T_u(B_u,R_u^S,F_u,\eta_u)$ and $V_u(B_u,R_u^S,F_u,\eta_u)$, referring to the fourth (or fifth) row of Table \ref{T_u V_u}, and the results are $T_u = \frac{D_u}{R_u(B_u)}$ and $V_u = 0$. This verifies the third row of Table \ref{T_u V_u eta_opt}.

In Fig. 8(c), when $F_u/\rho_u$ is relatively small (shown by line $L_6$), as $\eta_u$ increases from $0$ to $1$, it goes through zones Cond. 4, Cond. 1, Cond. 2, and Cond. 3 sequentially. Therefore, the optimal $\eta_u$ should be
\begin{align}
	\nonumber \eta_u^{\mathrm{opt}} & = \frac{F_u}{F_u(1-\zeta_u)+\rho_u R_u^S}, \\
	& \hspace{50pt} R_u(B_u) \geq \frac{R_u^S}{\zeta_u},\ \frac{F_u}{\rho_u} < \frac{R_u^S}{\zeta_u}.
\end{align}
We substitute this $\eta_u^{\mathrm{opt}}$ into the function expressions of $T_u(B_u,R_u^S,F_u,\eta_u)$ and $V_u(B_u,R_u^S,F_u,\eta_u)$, referring to the first (or second) row of Table \ref{T_u V_u}, and the results are $T_u = \frac{D_u \rho_u}{F_u(1-\zeta_u)+\rho_u R_u^S}$ and $V_u = \frac{D_u}{R_u(B_u)}\left[R_u(B_u) - R_u^S - (1-\zeta_u) \frac{F_u}{\rho_u}\right]$. This verifies the fourth row of Table \ref{T_u V_u eta_opt}.

When $F_u/\rho_u$ is larger (shown by line $L_7$), as $\eta_u$ increases from $0$ to $1$, it goes through zones Cond. 4 and Cond. 1 sequentially. During this process, $T_u$ decrease and decrease, whereas $V_u$ decrease and remain unchanged. Therefore, the optimal $\eta_u$ should be
\begin{equation}
	\eta_u^{\mathrm{opt}} = 1, \hspace{20pt} R_u(B_u) \geq \frac{R_u^S}{\zeta_u},\ \frac{R_u^S}{\zeta_u} \leq \frac{F_u}{\rho_u} < R_u(B_u).
\end{equation}
We substitute this $\eta_u^{\mathrm{opt}}$ into the function expressions of $T_u(B_u,R_u^S,F_u,\eta_u)$ and $V_u(B_u,R_u^S,F_u,\eta_u)$, referring to the first row of Table \ref{T_u V_u}, and the results are $T_u = \frac{\zeta_u D_u}{R_u^S}$ and $V_u = \frac{D_u}{R_u(B_u)}\left[R_u(B_u) - R_u^S - (1-\zeta_u) \frac{F_u}{\rho_u}\right]$. This verifies the fifth row of Table \ref{T_u V_u eta_opt}.

When $F_u/\rho_u$ is further increased (shown by line $L_8$), as $\eta_u$ increases from $0$ to $1$, it only goes through zone Cond. 4. During this process, both $T_u$ and $V_u$ decrease. Therefore, the optimal $\eta_u$ should be
\begin{equation}
	\eta_u^{\mathrm{opt}} = 1, \hspace{20pt} R_u(B_u) \geq \frac{R_u^S}{\zeta_u},\ \frac{F_u}{\rho_u} \geq R_u(B_u).
\end{equation}
We substitute this $\eta_u^{\mathrm{opt}}$ into the function expressions of $T_u(B_u,R_u^S,F_u,\eta_u)$ and $V_u(B_u,R_u^S,F_u,\eta_u)$, referring to the fourth row of Table \ref{T_u V_u}, and the results are $T_u = \frac{\zeta_u D_u}{R_u^S}$ and $V_u = \frac{D_u}{R_u(B_u)}(\zeta_u R_u(B_u) - R_u^S)$. This verifies the sixth row of Table \ref{T_u V_u eta_opt}. By now Theorem \ref{theo 0} has been proved.

\begin{table*}[t]
	\caption{Mapping between $(\mathbf{B}^{(1)}, \mathbf{R}^{S,(1)}, \mathbf{F}^{(1)})$ and $(\mathbf{B}^{(2)}, \mathbf{R}^{S,(2)}, \mathbf{F}^{(2)})$.}
	\centering
	\renewcommand{\arraystretch}{3.3}
	\begin{tabular}{|m{2.95in}|m{1.55in}|m{0.65in}|m{1.15in}|}
		\hline
		\textbf{Condition} & $B_u^{(2)}$ & $R_u^{S,(2)}$ & $F_u^{(2)}$ \\ \hline
		$R_u(B_u^{(1)}) < R_u^{S,(1)}$ & $B_u^{(1)}$ & $R_u(B_u^{(1)})$ & $0$ \\ \hline
		$R_u^{S,(1)} \leq R_u(B_u^{(1)}) < \frac{R_u^{S,(1)}}{\zeta_u}$, $\frac{F_u^{(1)}}{\rho_u} < \frac{R_u(B_u^{(1)}) - R_u^{S,(1)}}{1-\zeta_u}$  & $R_u^{-1}(R_u^{S,(1)}+\frac{(1-\zeta_u) F_u^{(1)}}{\rho_u})$ & $R_u^{S,(1)}$ & $F_u^{(1)}$ \\ \hline
		$R_u^{S,(1)} \leq R_u(B_u^{(1)}) < \frac{R_u^{S,(1)}}{\zeta_u}$, $\frac{F_u^{(1)}}{\rho_u} > \frac{R_u(B_u^{(1)}) - R_u^{S,(1)}}{1-\zeta_u}$  & $B_u^{(1)}$ & $R_u^{S,(1)}$ & $\rho_u\frac{R_u(B_u^{(1)})-R_u^{S,(1)}}{1-\zeta_u}$ \\ \hline
		$R_u(B_u^{(1)}) \geq \frac{R_u^{S,(1)}}{\zeta_u}$, $\frac{F_u^{(1)}}{\rho_u} < \frac{R_u^{S,(1)}}{\zeta_u}$ & $R_u^{-1}(R_u^{S,(1)}+\frac{(1-\zeta_u) F_u^{(1)}}{\rho_u})$ & $R_u^{S,(1)}$ & $F_u^{(1)}$ \\ \hline
		$R_u(B_u^{(1)}) \geq \frac{R_u^{S,(1)}}{\zeta_u}$, $\frac{F_u^{(1)}}{\rho_u} \geq \frac{R_u^{S,(1)}}{\zeta_u}$ & $R_u^{-1}(\frac{R_u^{S,(1)}}{\zeta_u})$ & $R_u^{S,(1)}$ & $\rho_u \frac{R_u^{S,(1)}}{\zeta_u}$ \\ \hline
	\end{tabular}
	\label{B R_S F 2}
\end{table*}

\section{Proof of Theorem \ref{theo 2}}
Since (P3) is (P2) reducing its feasible region by adding two constraints, it is obvious that $\mathbf{x}^*$ is a feasible solution to (P2). We adopt the method of proof by contradiction. Assume that $\mathbf{x}^*$ is an optimal solution to (P3), but not an optimal solution to (P2). This means there exists a feasible solution $\mathbf{x}^{(1)} = (B_{\rm total}^{(1)}, R^{S,(1)}_{\rm total}, F_{\rm total}^{(1)}, V_{\rm total}^{(1)}, \mathbf{B}^{(1)}, \mathbf{R}^{S,(1)}, \mathbf{F}^{(1)})$ to (P2) that could achieve a lower objective function value, namely
\begin{align}
	\nonumber & a_1 B_{\rm total}^{(1)} + a_2 R^{S,(1)}_{\rm total} + a_3 F_{\rm total}^{(1)} + a_4 V_{\rm total}^{(1)} \\
	& \hspace{10mm} < a_1 B_{\rm total}^* + a_2 R^{S,*}_{\rm total} + a_3 F_{\rm total}^* + a_4 V_{\rm total}^*. \label{theo 2 ineq}
\end{align}
Based on $\mathbf{x}^{(1)}$, we could derive a new feasible solution $\mathbf{x}^{(2)} = (B_{\rm total}^{(1)}, R^{S,(1)}_{\rm total}, F_{\rm total}^{(1)}, V_{\rm total}^{(1)}, \mathbf{B}^{(2)}, \mathbf{R}^{S,(2)}, \mathbf{F}^{(2)})$. The mapping between $(\mathbf{B}^{(1)}, \mathbf{R}^{S,(1)}, \mathbf{F}^{(1)})$ and $(\mathbf{B}^{(2)}, \mathbf{R}^{S,(2)}, \mathbf{F}^{(2)})$ are depicted in Table \ref{B R_S F 2}, referring to Appendix C of \cite{EIH RW sat-UAV MEC 05}. We note that here it is assumed $R_u(\cdot)$ is a strictly increasing function and $R_u^{-1}(\cdot)$ is its inverse function.

From Table \ref{B R_S F 2}, we could verify that $\mathbf{B}^{(2)} \leq \mathbf{B}^{(1)}$, $\mathbf{R}^{S,(2)} \leq \mathbf{R}^{S,(1)}$ and $\mathbf{F}^{(2)} \leq \mathbf{F}^{(1)}$. Therefore, $\mathbf{x}^{(2)}$ satisfies constraints (\ref{problem 1d})-(\ref{problem 1g}). Besides, we could refer to Table \ref{T_u V_u eta_opt} to verify that
\begin{align}
	\nonumber & T_u^{\eta-{\rm opt}}(B_u^{(2)},R_u^{S,(2)},F_u^{(2)}) = \frac{D_u}{R_u(B_u^{(2)})} \\
	& \hspace{15mm} \leq T_u^{\eta-{\rm opt}}(B_u^{(1)},R_u^{S,(1)},F_u^{(1)}), \ u \in \mathcal{U},
\end{align}
\begin{align}
	\nonumber & V_u^{\eta-{\rm opt}}(B_u^{(2)},R_u^{S,(2)},F_u^{(2)}) = 0 \\
	& \hspace{15mm} \leq V_u^{\eta-{\rm opt}}(B_u^{(1)},R_u^{S,(1)},F_u^{(1)}), \ u \in \mathcal{U},
\end{align}
which means $\mathbf{x}^{(2)}$ satisfies constraints (\ref{problem 2b}) and (\ref{problem 2c}). Moreover, it could be easily verified that $\mathbf{x}^{(2)}$ satisfies constraints (\ref{problem 3g}) and (\ref{problem 3h}) referring to Table \ref{B R_S F 2}. Consequently, $\mathbf{x}^{(2)}$ is a feasible solution to problem (P3). Considering inequality (\ref{theo 2 ineq}), $\mathbf{x}^*$ is not an optimal solution to (P3), causing a contradiction to the assumption. This proves the theorem.

\section{Proof of Theorem \ref{theorem B opt}}
We prove this theorem in two steps. First, we prove that there is no feasible solution to problem (\ref{problem 4}) that satisfies $B_{\rm total} < B_{\rm total}^*$. Since $B_{\rm total}^* = \sum_{u=1}^{U} B_u^*$, for a solution satisfying $B_{\rm total} < B_{\rm total}^*$, there exists at least one $u \in \mathcal{U}$ with $B_u < B_u^*$, based on the pigeonhole principle. Since $R_u(B_u)$ is strictly increasing, this equals to
\begin{equation}
	R_u(B_u) < R_u(B_u^*) = D_u/T_{\rm req}. \ u = 1,...,U.
\end{equation}
Therefore, this solution does not satisfy constraint (\ref{problem 4b}). This verifies that for any solution where $B_{\rm total} < B_{\rm total}^*$ holds, it is not in the feasible region of problem (\ref{problem 4}).

Then, we prove that for any feasible solution to problem (\ref{problem 4}) with $B_{\rm total} > B_{\rm total}^*$, we could find a feasible solution to problem (\ref{problem 4}), which satisfies (\ref{B_total opt}) and (\ref{B_u opt}), and obtains an equal or lower objective function value. Specifically, we consider a feasible solution $\mathbf{x}^{(1)} = (B_{\rm total}^{(1)}, R^{S,(1)}_{\rm total}, F^{(1)}_{\rm total}, \mathbf{B}^{(1)}, \mathbf{R}^{S,(1)}, \mathbf{F}^{(1)})$ with $B_{\rm total}^{(1)} > B_{\rm total}^*$. As a feasible solution, it satisfies that
\begin{equation}
	R_u(B_u^{(1)}) \geq  D_u/T_{\rm req} = R_u(B_u^*), \ u \in \mathcal{U},
\end{equation}
which equals to $B_u^{(1)} \geq B_u^*, \ u \in \mathcal{U}$. A new solution $\mathbf{x}^{(2)} = (B_{\rm total}^{(2)}, R^{S,(2)}_{\rm total}, F^{(2)}_{\rm total}, \mathbf{B}^{(2)}, \mathbf{R}^{S,(2)}, \mathbf{F}^{(2)})$ could be derived based on $\mathbf{x}^{(1)}$. The expressions of $\mathbf{x}^{(2)}$ could be given as follows
\begin{equation}
	B_u^{(2)} = B_u^*,\ u \in \mathcal{U}. \label{prf_theo3 B_u}
\end{equation}
\begin{equation}
	R_u^{S,(2)} = \frac{R_u(B_u^*)}{R_u(B_u^{(1)})} \cdot R_u^{S,(1)},\ u \in \mathcal{U}. \label{prf_theo3 R_u^S}
\end{equation}
\begin{align}
	\nonumber F_u^{(2)} & = \frac{R_u(B_u^*)}{R_u(B_u^{(1)})} \cdot F_u^{(1)} \\
	& = \frac{R_u(B_u^*)}{R_u(B_u^{(1)})} \cdot \rho_u \frac{R_u(B_u^{(1)})-R_u^{S,(1)}}{1-\zeta_u},\ u \in \mathcal{U}. \label{prf_theo3 F_u}
\end{align}
\begin{equation}
	B_{\rm total}^{(2)} = \sum_{u=1}^{U} B_u^{(2)} = B_{\rm total}^*. \label{prf_theo3 B_total}
\end{equation}
\begin{equation}
	R_{\rm total}^{S,(2)} = \sum_{u=1}^{U} R_u^{S,(2)}. \label{prf_theo3 R_total^S}
\end{equation}
\begin{equation}
	F_{\rm total}^{(2)} = \sum_{u=1}^{U} F_u^{(2)}. \label{prf_theo3 F_total}
\end{equation}
From the expressions we could easily verify that $\mathbf{x}^{(2)}$ is also a feasible solution to problem (\ref{problem 4}), and it satisfies (\ref{B_total opt}) and (\ref{B_u opt}). From (\ref{prf_theo3 B_u})-(\ref{prf_theo3 F_u}) we could verify that $B_u^{(2)} \leq B_u^{(1)}$, $R_u^{S,(2)} \leq R_u^{S,(1)}$ and $F_u^{(2)} \leq F_u^{(1)}$. Therefore, the following statement hold
\begin{equation}
	B_{\rm total}^{(2)} \leq \sum_{u=1}^{U} B_u^{(1)} \leq B_{\rm total}^{(1)}.
\end{equation}
\begin{equation}
	R_{\rm total}^{S,(2)} \leq \sum_{u=1}^{U} R_u^{S,(1)} \leq R_{\rm total}^{S,(1)}.
\end{equation}
\begin{equation}
	F_{\rm total}^{(2)} \leq \sum_{u=1}^{U} F_u^{(1)} \leq F_{\rm total}^{(1)}.
\end{equation}
Therefore, $\mathbf{x}^{(2)}$ could obtain equal or lower objective function value compared with $\mathbf{x}^{(1)}$.

In the case where $B_{\rm total} < B_{\rm total}^*$, there is no feasible solution to problem (\ref{problem 4}). In the case where $B_{\rm total} > B_{\rm total}^*$, for any feasible solution, we could derive another feasible solution that satisfies (\ref{B_total opt}) and (\ref{B_u opt}) and obtains equal or lower objective function value. Therefore, we verify that there exists at least one optimal solution to problem (\ref{problem 4}) that satisfies (\ref{B_total opt}) and (\ref{B_u opt}), which completes the proof of Theorem \ref{theorem B opt}.

\section{Proof of Theorem \ref{theorem R^S F opt}}
First, we could easily verify that the optimal solution to problem (\ref{problem 5}) must ensure that constraints (\ref{problem 5b}) and (\ref{problem 5c}) are tight. If they are not tight, we could obtain a better solution by simply lowering $R_{\rm total}^S$ or $F_{\rm total}$. Therefore, problem (\ref{problem 5}) could be equivalently transformed into the following problem
\begin{subequations} \label{problem apdx1}
	\begin{align}
		\mbox{(P10)}\ \min_{\mathbf{R}^S} &\ \ \sum_{u=1}^{U} \left[\left(a_2-a_3\frac{\rho_u}{1-\zeta_u}\right) R_u^S \right]+C \label{problem apdx1a} \\
		\mathrm{s.t.} &\ \ \zeta_u \frac{D_u}{T_{\rm req}} \leq R_u^S \leq \frac{D_u}{T_{\rm req}}, \ u \in \mathcal{U}, \label{problem apdx1b}
	\end{align}
\end{subequations}
where $C = \sum_{u=1}^{U} \frac{\rho_u}{1-\zeta_u} \frac{D_u}{T_{\rm req}} + a_1 B_{\rm total}^*$ is a constant. We could see that for a specific $u$, if $a_2-a_3\frac{\rho_u}{1-\zeta_u} \geq 0$, the objective function is increasing with $R_u^S$, and the optimal solution satisfies $R_u^S = \zeta_u \frac{D_u}{T_{\rm req}}$. Otherwise, when $a_2-a_3\frac{\rho_u}{1-\zeta_u} < 0$, the objective function is decreasing with $R_u^S$, and the optimal solution satisfies $R_u^S = \frac{D_u}{T_{\rm req}}$. To summarize, the optimal solution could be expressed as follows
\begin{equation}
	R_u^S = \frac{D_u}{T_{\rm req}} \left[ 1-(1-\zeta_u) \cdot {\rm I}\left( \frac{\rho_u}{1-\zeta_u} \leq \frac{a_2}{a_3} \right) \right],
\end{equation}
where ${\rm I}(\cdot)$ is an indicator function, which equals to 1 when the statement holds, and 0 if it does not hold.

We could further obtain the optimal $R_{\rm total}^S$ and $F_{\rm total}$ as follows
\begin{equation}
	R_{\rm total}^S = \sum_{u=1}^{U} \frac{D_u}{T_{\rm req}} \left[ 1-(1-\zeta_u) \cdot {\rm I}\left( \frac{\rho_u}{1-\zeta_u} \leq \frac{a_2}{a_3} \right) \right].
\end{equation}
\begin{equation}
	F_{\rm total} = \sum_{u=1}^{U} \rho_u \frac{D_u}{T_{\rm req}} \cdot {\rm I}\left( \frac{\rho_u}{1-\zeta_u} \leq \frac{a_2}{a_3} \right).
\end{equation}
This proves the theorem.

\section{Proof of Theorem \ref{theo placement prob equiv}}
First (\ref{problem 7c}) could be equivalently rewritten as
\begin{equation}
	\frac{p_u l_u^2}{\sigma^2 \nu_u} = \nu_u-1.
\end{equation}
Since (\ref{problem 7c}) holds in (P7), replacing (\ref{problem 7b}) with
\begin{equation}
	2\log\nu_u+\frac{1}{\nu_u}-1 = \frac{D_u\log2}{T_{\rm req} B_u}
\end{equation}
is an equivalent transformation.

Besides, referring to (\ref{large-scale channel})-(\ref{elevation angle}), (\ref{problem 7c}) equals to the following set of equations
\begin{align}
	\nu_u (\nu_u-1) & = \frac{p_u}{\sigma^2} \left(\frac{c}{4\pi f d_u}\right)^2 10^{\frac{\eta_{\rm NLoS}-\eta_{\rm LoS}}{10} Pr_u - \frac{\eta_{\rm NLoS}}{10}}, \label{problem apdx2c prev} \\
	d_u & = \sqrt{(h_x-h_{u,x})^2+(h_y-h_{u,y})^2+H^2}, \\
	Pr_u & = \frac{1}{1+a \cdot {\rm e}^{-b(\theta_u-a)}}, \\
	\theta_u & = \sin^{-1}(H/d_u),
\end{align}

Therefore, by introducing variables $\mathbf{Q}=\{Q_u, u \in \mathcal{U}\}$, $\mathbf{Pr}=\{Pr_u, u \in \mathcal{U}\}$ and $\bm{\theta}=\{\theta_u, u \in \mathcal{U}\}$, (P7) could be equivalently recast as
\begin{subequations} \label{problem apdx2}
	\begin{align}
		\mbox{(P11)} \min_{\substack{h_x, h_y, \\ \mathbf{B}, \bm{\nu}, \\ \mathbf{Q}, \mathbf{Pr}, \bm{\theta}}} \hspace{1mm}  & a_1 \sum_{u=1}^{U} B_u + a_2 R^{S,*}_{\rm total} + a_3 F^*_{\rm total} \label{problem apdx2a} \\
		\mathrm{s.t.} \hspace{3mm} & 2\log\nu_u+\frac{1}{\nu_u}-1 = \frac{D_u\log2}{T_{\rm req} B_u},\ u \in \mathcal{U}, \label{problem apdx2b} \\
		\nonumber & \log[\nu_u(\nu_u-1)] = \frac{\log 10}{10} (\eta_{\rm NLoS}-\eta_{\rm LoS}) Pr_u \\
		& \hspace{24mm} -\log Q_u + C,\ u \in \mathcal{U}, \label{problem apdx2c} \\
		\nonumber & Q_u = (h_x-h_{u,x})^2+(h_y-h_{u,y})^2 \\
		& \hspace{35mm} +H^2 ,\ u \in \mathcal{U}, \label{problem apdx2d} \\
		& \frac{1}{Pr_u} = 1+a \cdot {\rm e}^{-b(\theta_u-a)},\ u \in \mathcal{U}, \label{problem apdx2e} \\
		& \sin \theta_u = \frac{H}{\sqrt{Q_u}},\ u \in \mathcal{U}, \label{problem apdx2f} \\
		& \nu_u \geq 1,\ u \in \mathcal{U}, \label{problem apdx2g}
	\end{align}
\end{subequations}
where $C = \log(p_u/\sigma^2)+2\log(c/4\pi f)-(\log 10/10) \cdot \eta_{\rm NLoS}$, and thus (\ref{problem apdx2c}) is actually (\ref{problem apdx2c prev}) taking logarithm on both sides.

Comparing (P8) with (P11), we could see that if equality holds in constraints (\ref{problem 8b})-(\ref{problem 8f}) in the optimal solution, (P8) is equivalent to (P11) and therefore equivalent to (P7). First, (\ref{problem 8b}) must hold in the optimal solution since (P8) aims to minimize $B_u$. 

Define function $f(x) = 2\log x+1/x-1$. From (\ref{problem 8b}) we know that in order to minimize $B_u$ we need to maximize $f_1(\nu_u)$. It could be easily verified that $f(x)$ is monotonically increasing when $x \geq 1$. Due to this monotonicity, $\nu_u$ is maximized in the optimal solution, and thus equality must hold in (\ref{problem 8c}).

From (\ref{problem 8c}) we know that in order to maximize $\nu_u$ we need $Pr_u$ as large as possible while $Q_u$ as small as possible. Since $Pr_u$ is maximized, equality must hold in (\ref{problem 8e}).

Besides, in order to maximize $Pr_u$, $\theta_u$ should be as large as possible, and therefore equality must hold in (\ref{problem 8f}).

Moreover, maximizing $\theta_u$ also requires $Q_u$ to be as small as possible referring to (\ref{problem 8f}). Thus, equality must hold in (\ref{problem 8d}) in the optimal solution.

By now we have proven that equality holds in constraints (\ref{problem 8b})-(\ref{problem 8f}) in the optimal solution. This proves that (P8) is equivalent to (P11) and therefore equivalent to (P7), which proves the theorem.

\vspace{10mm}

%\begin{thebibliography}{1}
%	\bibliographystyle{IEEEtran}

\vfill


\begin{thebibliography}{10}
	\providecommand{\url}[1]{#1} \csname url@rmstyle\endcsname
	\providecommand{\newblock}{\relax}
	\providecommand{\bibinfo}[2]{#2}
	\providecommand\BIBentrySTDinterwordspacing{\spaceskip=0pt\relax}
	\providecommand\BIBentryALTinterwordstretchfactor{4}
	\providecommand\BIBentryALTinterwordspacing{\spaceskip=\fontdimen2\font
		plus \BIBentryALTinterwordstretchfactor\fontdimen3\font minus
		\fontdimen4\font\relax}
	\providecommand\BIBforeignlanguage[2]{{%
			\expandafter\ifx\csname l@#1\endcsname\relax
			\typeout{** WARNING: IEEEtran.bst: No hyphenation pattern has been}%
			\typeout{** loaded for the language `#1'. Using the pattern for}%
			\typeout{** the default language instead.}%
			\else \language=\csname l@#1\endcsname \fi #2}}
	
	\bibitem{intro 01}
	Statista. (Dec. 6, 2024). {\it Economic Loss from Natural Disaster Events Worldwide in 2023, by Peril (in Billion U.S. Dollars)}. [Online]. Available: https://www.statista.com/statistics/510922/natural-disasters-globally-and-economic-losses-by-peril/
	
	\bibitem{intro 02}
	T. Kedia, J. Ratcliff, M. O'Connor, {\it et al.}, ``Technologies enabling situational awareness during disaster response: a systematic review,” {\it Disaster Med. Public Health Prep.}, vol. 16, no. 1, pp. 341–359, 2022.
	
	\bibitem{intro 03}
	D. G. C., A. Ladas, Y. A. Sambo, H. Pervaiz, C. Politis, and M. A. Imran, ``An overview of post-disaster emergency communication systems in the future networks,” {\it IEEE Wireless Commun.}, vol. 26, no. 6, pp. 132–139, Dec. 2019.
	
	\bibitem{intro 04}
	M. Matracia, N. Saeed, M. A. Kishk, and M. -S. Alouini, ``Post-disaster communications: Enabling technologies, architectures, and open Challenges," {\it IEEE Open J. Commun. Soc.}, vol. 3, pp. 1177-1205, 2022.
	
	\bibitem{intro 05}
	W. Feng, Y. Wang, Y. Chen, N. Ge, and C.-X. Wang, ``Structured satellite-UAV-terrestrial networks for 6G Internet of Things,” {\it IEEE Netw.}, vol. 38, no. 4, pp. 48-54, Jul. 2024.
	
	\bibitem{intro s01}
	T. Wei, W. Feng, Y. Chen, C.-X. Wang, N. Ge, and J. Lu, ``Hybrid satellite-terrestrial communication networks for the maritime Internet of Things: Key technologies, opportunities, and challenges,” {\it IEEE Internet Things J.}, vol. 8, no. 11, pp. 8910-8934, Jun. 2021.
	
	\bibitem{NTN}
	M. Giordani and M. Zorzi, ``Non-terrestrial networks in the 6G era: Challenges and opportunities," {\it IEEE Netw.}, vol. 35, no. 2, pp. 244-251, Mar./Apr. 2021.
	
	\bibitem{intro s02}
	W. Feng, Y. Lin, Y. Wang, {\it et al.}, ``Radio map-based cognitive satellite-UAV networks towards 6G on-demand coverage,” {\it IEEE Trans. Cogn. Commun. Netw.}, vol. 10, no. 3, pp. 1075-1089, Jun. 2024.
	
	\bibitem{MEC}
	N. Abbas, Y. Zhang, A. Taherkordi, and T. Skeie, ``Mobile edge computing: A survey," {\it IEEE Internet Things J.}, vol. 5, no. 1, pp. 450-465, Feb. 2018.
	
	\bibitem{EIH RW sat-UAV MEC 05}
	Y. Lin, W. Feng, Y. Chen, N. Ge, Z. Feng, and Y. Gao, ``Edge information hub-empowered 6G NTN: Latency-oriented resource orchestration and configuration," \textit{IEEE Open J. Commun. Soc.}, vol. 5, pp. 4241-4259, 2024.
	
	\bibitem{RW NTN 01}
	X. Fang, W. Feng, T. Wei, Y. Chen, N. Ge, and C. -X. Wang, ``5G embraces satellites for 6G ubiquitous IoT: Basic models for integrated satellite terrestrial networks," {\it IEEE Internet Things J.}, vol. 8, no. 18, pp. 14399-14417, Sep. 2021.
	\bibitem{RW NTN s01}
	Z. Wei, M. Zhu, N. Zhang, {\it et al.}, ``UAV-assisted data collection for Internet of Things: A survey,” {\it IEEE Internet Things J.}, vol. 9, no. 17, pp. 15460–15483, Sep. 2022.
	\bibitem{RW NTN s02}
	Y. Bian, J. Hu, P. Zhang, {\it et al.}, ``Joint trajectory control, power control, and collection schedule in UAV-assisted anti-jamming wireless data collection with imperfect CSI," {\it IEEE Commun. Lett.}, vol. 28, no. 12, pp. 2839-2843, Dec. 2024.
	\bibitem{RW NTN s03}
	M. Li, S. He, and H. Li, ``Minimizing mission completion time of UAVs by jointly optimizing the flight and data collection trajectory in UAV-enabled WSNs," {\it IEEE Internet Things J.}, vol. 9, no. 15, pp. 13498-13510, Aug. 2022.

	%\bibitem{RW NTN 02}
	%J. Wang, Z. Na, and X. Liu, ``Collaborative design of multi-UAV trajectory and resource scheduling for 6G-enabled Internet of Things," {\it IEEE Internet Things J.}, vol. 8, no. 20, pp. 15096-15106, Oct. 2021.
	
	\bibitem{RW NTN 03}
	C. Liu, W. Feng, Y. Chen, C. -X. Wang, and N. Ge, ``Cell-free satellite-UAV networks for 6G wide-area Internet of Things," {\it IEEE J. Sel. Areas Commun.}, vol. 39, no. 4, pp. 1116-1131, Apr. 2021.
	
	\bibitem{RW NTN 04}
	Z. Li, Y. Wang, M. Liu, {\it et al.}, ``Energy efficient resource allocation for UAV-assisted space-air-ground Internet of Remote Things networks," {\it IEEE Access}, vol. 7, pp. 145348-145362, 2019.
	
	\bibitem{RW NTN 05}
	T. Ma, H. Zhou, B. Qian, {\it et al.}, ``UAV-LEO integrated backbone: A ubiquitous data collection approach for B5G Internet of Remote Things networks," {\it IEEE J. Sel. Areas Commun.}, vol. 39, no. 11, pp. 3491-3505, Nov. 2021.
	
	\bibitem{RW UAV MEC 01}
	J. Zhang, G. Zhang, X. Wang, X. Zhao, P. Yuan, and H. Jin, ``UAV-assisted task offloading in edge computing," {\it IEEE Internet Things J.}, vol. 12, no. 5, pp. 5559-5574, Mar. 2025.
	
	\bibitem{RW UAV MEC 02}
	L. Wang, Y. Li, Y. Chen, T. Li, and Z. Yin, ``Air–ground coordinated MEC: Joint task, time allocation and trajectory design," {\it IEEE Trans. Veh. Technol.}, vol. 74, no. 3, pp. 4728-4743, Mar. 2025.
	
	\bibitem{RW UAV MEC 03}
	J. Zhang, H. Luo, X. Chen, H. Shen, and L. Guo, ``Minimizing response delay in UAV-assisted mobile edge computing by joint UAV deployment and computation offloading," {\it IEEE Trans. Cloud Comput.}, vol. 12, no. 4, pp. 1372-1386, Oct.-Dec. 2024.
	
	\bibitem{RW UAV MEC 04}
	L. Wang, K. Wang, C. Pan, W. Xu, N. Aslam, and A. Nallanathan, ``Deep reinforcement learning based dynamic trajectory control for UAV-assisted mobile edge computing," {\it IEEE Trans. Mobile Comput.}, vol. 21, no. 10, pp. 3536-3550, Oct. 2022.
	
	\bibitem{RW UAV MEC 05}
	Y. Liu, C. Yang, Y. Tang, H. Zhao, Y. Liu, and S. Xie, ``Cost-efficient deployment optimization for multi-UAV-assisted vehicular edge computing networks," {\it IEEE Internet Things J.}, vol. 12, no. 6, pp. 6158-6170, Mar. 2025.
	
	\bibitem{RW UAV MEC 06}
	A. M. Seid, G. O. Boateng, S. Anokye, T. Kwantwi, G. Sun, and G. Liu, ``Collaborative computation offloading and resource allocation in multi-UAV-assisted IoT networks: A deep reinforcement learning approach," {\it IEEE Internet Things J.}, vol. 8, no. 15, pp. 12203-12218, Aug. 2021.
	
	\bibitem{RW UAV MEC 07}
	N. Zhao, Z. Ye, Y. Pei, Y. -C. Liang, and D. Niyato, ``Multi-agent deep reinforcement learning for task offloading in UAV-assisted mobile edge computing," {\it IEEE Trans. Wireless Commun.}, vol. 21, no. 9, pp. 6949-6960, Sep. 2022.
	
	\bibitem{RW sat-UAV MEC 01}
	Y. Lin, W. Feng, T. Zhou, {\it et al.}, ``Integrating satellites and mobile edge computing for 6G wide-area edge intelligence: Minimal structures and systematic thinking," {\it IEEE Netw.}, vol. 37, no. 2, pp. 14-21, Mar./Apr. 2023.
	
	\bibitem{RW sat-UAV MEC 02}
	C. Zhou, W. Wu, H. He, {\it et al.}, ``Deep reinforcement learning for delay-oriented IoT task scheduling in SAGIN," {\it IEEE Trans. Wireless Commun.}, vol. 20, no. 2, pp. 911-925, Feb. 2021.
	
	\bibitem{RW sat-UAV MEC 03}
	Y. K. Tun, K. T. Kim, L. Zou, Z. Han, G. Dán, and C. S. Hong, ``Collaborative computing services at ground, air, and space: An optimization approach," {\it IEEE Trans. Veh. Technol.}, vol. 73, no. 1, pp. 1491-1496, Jan. 2024.
	
	\bibitem{RW sat-UAV MEC 04}
	N. N. Ei, J. S. Yoon, and C. S. Hong, ``Energy-aware task offloading and resource allocation in space-aerial-integrated MEC system," in {\it Proc. Asia-Pac. Netw. Oper. Manag. Symp. (APNOMS)}, Takamatsu, Japan, 2022, pp. 1-6.
	
	\bibitem{RW sat-UAV MEC 06}
	Y.-H. Chao, C.-H. Chung, C.-H. Hsu, Y. Chiang, H.-Y. Wei, and C.T. Chou, ``Satellite-UAV-MEC collaborative architecture for task offloading in vehicular networks," in {\it Proc. IEEE Globecom Workshops (GC Wkshps)}, Taipei, Taiwan, 2020, pp. 1-6.
	
	\bibitem{RW sat-UAV MEC 07}
	S. Jung, S. Jeong, J. Kang, and J. Kang, ``Marine IoT systems with space–air–sea integrated networks: Hybrid LEO and UAV edge computing," {\it IEEE Internet Things J.}, vol. 10, no. 23, pp. 20498-20510, Dec. 2023.
	
	\bibitem{RW sat-UAV MEC 08}
	S. Qi, B. Lin, Y. Deng, X. Chen, and Y. Fang, ``Minimizing maximum latency of task offloading for multi-UAV-assisted maritime search and rescue," {\it IEEE Trans. Veh. Technol.}, vol. 73, no. 9, pp. 13625-13638, Sep. 2024.
	
	\bibitem{RW sat-UAV MEC 09}
	P. Tong, J. Liu, X. Wang, B. Bai, and H. Dai, ``UAV-Enabled age-optimal data collection in wireless sensor networks," in {\it Proc. IEEE Int. Conf. Commun. Workshops (ICC Wkshps)}, Shanghai, China, 2019, pp. 1-6.
	
	\bibitem{RW sat-UAV MEC 10}
	C. Huang, G. Chen, P. Xiao, Y. Xiao, Z. Han, and J. A. Chambers, ``Joint offloading and resource allocation for hybrid cloud and edge computing in SAGINs: A decision assisted hybrid action space deep reinforcement learning approach," {\it IEEE J. Sel. Areas Commun.}, vol. 42, no. 5, pp. 1029-1043, May 2024.
	
%	\bibitem{syst}
%	W. Feng, Y. Wang, Y. Chen, N. Ge, and C. -X. Wang, ``Structured satellite-UAV-terrestrial networks for 6G Internet of Things," {\it IEEE Netw.}, vol. 38, no. 4, pp. 48-54, Jul. 2024.
	
	\bibitem{channel}
	A. Al-Hourani, S. Kandeepan, and S. Lardner, ``Optimal LAP altitude	for maximum coverage,” {\it IEEE Wireless Commun. Lett.}, vol. 3, no. 6, pp. 569–572, Dec. 2014.
	
	\bibitem{ergodic rate}
	W. Feng, Y. Wang, N. Ge, J. Lu, and J. Zhang, ``Virtual MIMO in multi-cell distributed antenna systems: Coordinated transmissions with large-scale CSIT,” {\it IEEE J. Sel. Areas Commun.}, vol. 31, no. 10, pp. 2067-2081, Oct. 2013.
	\bibitem{channel coeff}
	Y. Chen, W. Feng, and G. Zheng, ``Optimum placement of UAV as relays," {\it IEEE Commun. Lett.}, vol. 22, no. 2, pp. 248-251, Feb. 2018.
	
\end{thebibliography}
\end{document}